\documentclass[12pt]{article}

\usepackage[utf8]{inputenc}

\usepackage[longnamesfirst]{natbib}
\usepackage{mathtools, amsfonts, amssymb,amsthm,dsfont, color, xcolor, tikz}
\usetikzlibrary{shapes.geometric, arrows, positioning, fit, calc, shapes.misc} % necessary to draw diagrams in tikz
\mathtoolsset{showonlyrefs}

\usepackage{enumerate, comment, subcaption, graphicx, wrapfig}
\usepackage{eurosym} % for the euro symbol
\usepackage{ifthen}
\usepackage{hyperref}
\usepackage[ruled,vlined]{algorithm2e}
\usepackage{prodint}

\usepackage{bbm}

\usepackage{xr}
\usepackage[inner=30mm,outer=20mm,top=30mm,bottom=20mm]{geometry} % from the doctoral school
\newcounter{prop}
\newtheorem{proposition}[prop]{Proposition}
\newcounter{cor}
\newtheorem{corollary}[cor]{Corollary}
\newcounter{lem}
\newtheorem{lemma}[lem]{Lemma}

\newcounter{rem}
\newtheorem{remark}[rem]{Remark}

\newtheorem{assum}{Assumption}

\newcommand{\EE}{\mathbb{E}}

\newcommand{\PP}{\mathbb{P}}

\newcommand{\1}{\mathbbm{1}}

\newcommand{\x}{{\boldsymbol x}}
\newcommand{\ub}{{\boldsymbol u}}
\newcommand{\X}{{\textbf{X}}}
\newcommand{\bbeta}{{\boldsymbol \beta}}

\newcommand{\z}{\boldsymbol z}
\newcommand{\Z}{\boldsymbol Z}

\author{Daphn\'e Aurouet\footnote{Ensai, CREST - UMR 9194, France; daphne.aurouet@ensai.fr} \qquad Valentin Patilea\footnote{Corresponding authors: Ensai, CREST - UMR 9194, France; valentin.patilea@ensai.fr}}

\date{\today}
\title{Continuously updated  estimation of conditional hazard functions}

\begin{document}
	\maketitle
	
\begin{abstract}
Motivated by the need to analyze continuously updated data sets in the context of time-to-event modeling, we propose a novel nonparametric approach to estimate the conditional hazard function given a set of continuous and discrete predictors. The method is based on a representation of the conditional hazard as a ratio between a joint density and a conditional expectation determined by the distribution of the observed variables. It is shown that such ratio representations are available for uni- and bivariate time-to-events, in the presence of common types of random censoring, truncation, and with possibly cured individuals, as well as for competing risks. This opens the door to nonparametric approaches in many time-to-event predictive models. To estimate joint densities and conditional expectations we propose the recursive kernel smoothing, which is well suited for online estimation. Asymptotic results for such estimators are derived and it is  shown that they achieve optimal convergence rates. Simulation experiments show the good finite sample performance of our recursive estimator with right censoring. The method is applied to a real dataset of primary breast cancer.
\end{abstract}

\medskip



%\localtableofcontents

%-------------------------------
\newpage \section{Introduction}
%-------------------------------
% !TeX root = ../hazard.tex

Time-to-event analysis, also called event history, is a long-standing topic in statistics that deals with the time to a particular event or outcome. Application fields include, but are not limited to, biostatistics, engineering, labour economics, finance and insurance, and demography. In the most common setting, the variable of interest is a non-negative random variable $T$ representing the time until the event of interest occurs. The aim is to estimate the distribution of this variable, or its conditional distribution if predictors are also observed.

The main challenge in time-to-event analysis comes from the fact that the variable $T$ is not always observed. For a fraction, or all, of the observed individuals, the data only indicate that $T$ is less than or greater than some observed value, or that $T$ is between two observed values. The existence of some individuals may even be unknown, or the value of their time-to-event $T$ is infinite. We shall  call all reasons that prevent the exact observation of $T$ \emph{incompleteness mechanisms}. Thus, they include censoring (left, right, interval), truncation (left, right), and long-time survivors (cured individuals). See, the textbooks and reviews \cite{KM2003}, \cite{BR2004}, \cite{KK2012}, \cite{KvHIS2014}, \cite{AVK2018} for extensive lists of different types of common incompleteness mechanisms in time-to-event analysis. 

Estimating the distribution of $T$ is often done through the hazard function, especially when predictors are available, in which case the interest is in the conditional hazard function given the predictors (also called covariates). 
The classical approaches to modeling the hazard function are parametric and semiparametric as in the case of the famous Cox's proportional hazard model. However, model-free (nonparametric) approaches are preferable when the underlying assumptions of the (semi)parametric models do not meet consensus. Model-free approaches can be effective as soon as data sets of hundreds of individuals are available.

A key message of this paper is that the conditional hazard function allows for a convenient representation, as a ratio between the joint density and a conditional expectation of the observed variables. This representation is well known in the case of right-censoring, see for example  \cite{E2018}, \cite{E2024}. We recall that similar representations exist in other situations, and show that it can be derived in many others, yet unexplored contexts where an incompleteness mechanism prevents the observation of $T$ for all individuals. Our examples include models with left or right random censoring,  with or without random truncation, with or without a cure fraction, and also some types of modified current status models. The cases with competing risks or bivariate time-to-event in the presence of random right censoring are also discussed. By integrating the hazard function and applying the product integral, the hazard representation allow the conditional survival of $T$ to be expressed as a closed-form function of the distribution of the data.

If such ratio-type representations of conditional hazard can be established, it opens the door to the use of nonparametric estimators of the numerator and denominator. In particular those adapted to massive data sets, possibly  arriving in streams. We propose to use recursive kernel estimators. Based on the stochastic approximation idea, they are easy to implement and to update, and achieve the optimal rate of convergence. 

Recursive kernel density estimation for fully observed data was  introduced by \cite{WW1969} and \cite{Y1971}, alternative weighting schemes can be found in \cite{WD1979} and \cite{D2013}. For the Wolverton-Wagner-Yamato estimator, theoretical convergence guarantees can be found in  \cite{D1972}
%, \cite{I1982} 
and \cite{R1992}. For the choice of the bandwidth, one can refer to \cite{KX2019} when the regularity is given, or to \cite{CM2020} for adaptive strategies. For extensions of the recursive approach to nonparametric kernel regression, see, among many others, 
%\cite{R1977}, 
\cite{R1992}, \cite{D2013}, \cite{KX2019}. Extensions of the recursive kernel estimators for density or regression in the context of right censored data have been  considered using the idea of inverse probability censoring weighting (IPCW)  with the Kaplan-Meier  estimator of the censoring survival function; see for example \cite{KS2014} and \cite{SK2020}. The IPCW approach in the recursive context is limited by the fact that the  Kaplan-Meier  estimator is computed using the whole sample, and in the case of a regression problem it requires the censoring independent of the predictors. Recursive hazard function estimation has been considered by \cite{R1992} and \cite{DMP2018}. The recursive estimation of the conditional hazard have received limited attention, even in the  right censoring case.  

The paper is organized as follows. In Section \ref{sec:chap1-model-methodo}, assuming  $T$ and the censoring  independent given the predictors, we derive the conditional hazard function representation in the case of right censoring; see also \cite{E2024}. The recursive kernel estimation of the numerator and the denominator in the hazard representation under right censoring is introduced in Section \ref{sec:chap1-model-estimation}. It allows for both continuous and discrete predictors. In Section \ref{sec:chap1-model-discussion}, we show that similar conditional hazard representations, a ratio between a joint density and a conditional expectation, can be derived with many other  incompleteness mechanisms. The same type or recursive nonparametric estimators can then be applied. We thus  provide a novel unifying perspective on model-free hazard estimation in predictive time-to-event models. The section ends with a discussion of the representations of the conditional cumulative hazard, leading to the conditional survival functions via product integration. In Section \ref{sec:chap1-model-theory}, we study the rate of uniform convergence in probability of a generic kernel estimator, which can be particularized to obtain recursive estimators for the numerator and the denominator in the different representations that were discussed. Our generic estimator is shown to achieve the optimal rate of uniform convergence on compact intervals, which implies the optimal rate of uniform convergence for the conditional hazard function. In Section \ref{sec:chap1-emp-analysis}, we present  a simulation study in the case of right censoring and the analysis of a real data set presented in \cite{RA2013} on patients suffering from primary breast cancer. Our model-free estimator of the conditional hazard has good performance and is easy to compute with large samples. Few technical details are gathered in an Appendix, additional details are given in a Supplement.

%--------------------------------------------------------------
\section{Methodology}\label{sec:chap1-model-methodo}
%--------------------------------------------------------------
% !TeX root = ../hazard.tex

In our framework, the data include a continuous random variable  $T\geq 0 $,  representing the time to the event of interest, and a vector $\X = (\X_c, \X_d)$ of covariates (also called predictors). The sub-vector $\X_c \in \mathcal X_c \subset  \mathbb R^{d_c}$ contains $d_c\geq 1$ continuous predictors, and $\X_d\in \mathcal X_d$ gathers the discrete ones. Let $\mathcal{X}$  denote the support of $\X$ and $\x=(\x_c,\x_d)$ denote a point in $\mathcal{X}$.

To account for right-censoring, we consider another random positive variable $C$,  the censoring time. Instead of $T$, the observed variables are the duration $Y=T \wedge C$ together with  the indicator $\delta = \1\{T \leq C\}$. The data consist of independent copies of $(Y, \delta, \X)$.

%\medskip

\begin{assum}\label{id_ass}
 $T$ and $C$ are conditionally independent given $\X$ (notation~: $T\perp C\mid \X$). 
\end{assum} 

%\medskip

Assumption \ref{id_ass} is the common identification condition.  Next, let 
$$
F_{T\mid \X}(t\mid\x) = \PP (T\leq t \mid \X=\x)\quad \text{ and } \quad F_{C\mid \X}(t\mid\x) = \PP (C\leq t \mid \X=\x),\quad t\geq 0,\x\in\mathcal{X},
$$
be the conditional distribution functions of $T$ and $C$ given $\X=\x$, respectively. Let
$S_{T\mid \X} = 1- F_{T\mid \X}$   and   $S_{C\mid \X} = 1- F_{C\mid \X}$ denote the corresponding survival functions. 
For each $\x\in\mathcal{X}$, we also consider the conditional sub-distributions functions of $Y$ given $\X=\x$, that are
\begin{equation}\label{eq:def_Hk}
	H_k(t\mid \x) = \mathbb{P}(Y \leq t, \delta = k\mid \X=\x), \qquad k\in\{0,1\}, %, \; t\geq 0,
\end{equation}
and the conditional distribution function 
%\begin{equation}\label{eq:def_H}
	$H(t| \x) = \mathbb{P}(Y \leq t| \X=\x) = H_0(t| \x) + H_1(t| \x)$. %, \qquad t\geq 0. 
%\end{equation}
By Assumption \ref{id_ass}, we have 
\begin{equation}\label{prod_STSC}
	1-H(t\mid \x) = S_{T\mid \X}(t\mid\x)S_{C\mid \X}(t\mid\x) , 
\end{equation}
and the following relationships hold true~: 
\begin{equation}\label{eq:sys1}
	\left\{\begin{array}{ll}
		H_1(dt\mid \x) = S_{C\mid \X}(t-\mid \x) F_{T\mid \X}(dt\mid \x) \\
		H_0(dt\mid \x) = S_{T\mid \X}(t\mid \x) F_{C\mid \X}(dt\mid \x)
	\end{array}\right. ,\qquad t\geq 0,\x\in\mathcal X.
\end{equation}
Here, for any $\x \in \mathcal X$,  $H_k(dt\mid \x)$, $F_{T\mid \X}(dt\mid \x)$ and $F_{C\mid \X}(dt\mid \x)$ are the measures associated to the càdlàg functions $H_k(\cdot\mid \x)$, $F_{T\mid \X}(\cdot\mid \x)$ and $F_{C\mid \X}(\cdot\mid \x)$, respectively. 
The system \eqref{eq:sys1} defines a one-to-one mapping between the `latent' conditional distributions $F_{T\mid \X}$, $F_{C\mid \X}$ and 
the observed conditional distribution defined by $H_0$ and $H_1$. Solving this system allows to express the cumulative hazard function of $T$, and next $F_{T\mid \X}$, as an explicit function of $H_0$ and $H_1$. Plugging standard kernel estimators for $H_0$ and $H_1$ in this map leads to the \emph{conditional Kaplan-Meier} estimator. See \cite{B1981}, \cite{D1987}, \cite{D1989}. 

Here we follow the idea of expressing a quantity of interest as a function of quantities that can be consistently estimated with censored data, with a focus on the conditional hazard function of $T$ given the predictors. For any $\x \in \mathcal X$, let $f_{T\mid \X}(t\mid \x)$ denote the conditional density of $T$ given $\X = \x$, such that $F_{T\mid \X}(dt\mid \x)= f_{T\mid \X}(t\mid \x)dt$. Let $f_{Y,\delta}(t,1\mid \x)$ denote the sub-density corresponding to $H_1(dt\mid \x)$, that means 
$
H_1(dt\mid \x)  = f_{Y,\delta\mid \X}(t,1\mid \x) dt,
$
where $f_{Y,\delta\mid \X}(\cdot,\cdot\mid \x)$ denotes the conditional joint density of $(Y,\delta)$, given $\X=\x$, with respect to the product measure between the Lebesgue measure and the counting measure. Using the first equation in \eqref{eq:sys1} and  \eqref{prod_STSC}, we write the conditional hazard  function given $\X=\x$ as
\begin{equation}\label{eq:haz_1}
	\lambda_{T\mid \X}(t\mid \x) = \frac{f_{T\mid \X}(t\mid \x)}{S_{T\mid \X}(t\mid \x)} =  \frac{f_{Y,\delta\mid \X}(t,1\mid \x)}{1 - H(t-\mid \x)}. 
\end{equation}
As a consequence, for the conditional survival function of $T$ we get
\begin{equation}\label{eq:S_1}
	S_{T\mid \X}(t\mid \x)= \exp\left(- \int_0^t 	\lambda_{T\mid \X}(s\mid \x)  ds\right) = \exp\left(- \int_0^t \frac{f_{Y,\delta\mid \X}(s,1\mid \x)}{1 - H(s\mid \x)} ds\right).
\end{equation}

In order to define our recursive estimator, we rewrite the conditional hazard function $\lambda_{T\mid \X}$ under a more convenient form.
Let $f_\X(\x)$ denotes the  density of $\X$ with respect to the product measure between the Lebesgue measure and the counting measure. Then,  
\begin{equation}\label{eq:h1_equiv22_a}
f_{Y,\delta\mid \X}(t,1\mid \x)f_\X(\x) = f_{Y,\delta,\X}(t,1,\x)= f_{Y,\X\mid \delta}(t, \x \mid \delta = 1)  \mathbb{P}(  \delta = 1),
\end{equation}
with $f_{Y,\delta,\X}$ the joint density of the data.
As a consequence, the conditional hazard is written as 
\begin{equation}\label{eq:haz_2}
	\lambda_{T\mid \X}(t\mid \x)  
	= \frac{f_{Y,\delta,\X}(t,1,\x )}   {R(t,\x)} \qquad \text{with } \quad R(t,\x)=\mathbb E(\1\{Y\geq t \} \mid \X=\x) f_\X(\x) .
	% = \frac{f_{Y,\X\mid \delta}(t,\x\mid \delta=1)}   {R(t,\x)} \mathbb{P}(  \delta = 1 ),  
\end{equation}
%where 
%\begin{equation}\label{eq:def-R}
%	R(t,\x)=\mathbb E(\1\{Y\geq t \} \mid \X=\x) f_\X(\x) .
%\end{equation}
See also \citep[][Eq. (6.2.5)]{E2018} for the unconditional version of this representation of the hazard function. The function $R(t,\x)$ is called the \emph{natural nuisance function} in \cite{E2024} and plays a key role in the minimax theory of nonparametric estimation.
% of the conditional hazard function. 

%--------------------------------------------------------------
\section{Estimation}\label{sec:chap1-model-estimation}
%--------------------------------------------------------------
% !TeX root = ../hazard.tex

The estimation approach is based on the identity \eqref{eq:haz_2} and recursive, nonparametric procedures for the numerator and the denominator.   For this purpose, we introduce some notation~:   
$$ 
\boldsymbol K_h (\ub)= h^{-d_c} \boldsymbol K (\ub/h), \quad \text{ where  } \quad \boldsymbol K: \mathbb R^{d_c} \rightarrow [0,\infty),
$$ 
is a multivariate kernel function, for example a product of univariate densities, and $h>0$ is a bandwidth. Similarly,  $L_b (y) = b^{-1}L(y/b)$, where $L: \mathbb R  \rightarrow [0,\infty)$ is a univariate kernel and $b$ a bandwidth. The sample  $(Y_i, \delta_i, \X_i)$, $i=1, \dots,n$, are independent copies of $(Y, \delta, \X)$, and the covariate vectors $\X_i$ are split into continuous and discrete components as $\X_i=(\X_{i,c},\X_{i,d})$.

Using the recursive kernel density estimator idea \citep{WW1969}  we first build the estimator of the joint density
$f_{Y,\delta,\X}(t,1,\x )$~: for $t\geq 0$, $\x=(\x_c,\x_d)$, 
\begin{multline}\label{eq:fRNWuni}
	\widehat{f}_{i} (t, 1, \x) = \frac{\delta_i}{i} L_{b(i)}\left(Y_i - t\right) \boldsymbol K_{h(i)}\left(\X_{i,c} - \x_c\right) \1 \{  \X_{i,d} = \x_d\}  \\
	+ \frac{i-1}{i} \widehat{f} _{i-1}(t,1,\x),\quad i\geq 1,\qquad \text{ with} \quad \widehat{f} _{0}(\cdot,1,\cdot)=0.
\end{multline}
Here, $b(i)$ and $h(i)$, $i\geq 1$, are sequences of bandwidth decreasing to zero at suitable rates. Regarding the recursive estimator of $R(t,\x)$, 
%given that the denominator of the Nadaraya-Watson estimator is the density of the predictor variables, 
a natural choice 
is~:  $\forall i \geq 1$, 
\begin{equation}\label{eq_rec_Hf}
	\widehat R_i(t,\x)  = \frac{\1 \{Y_{i}\geq t \}  }{i }  \boldsymbol  K_{g(i)} \left(  \X_{i,c} - \x_c\right)\1 \{  \X_{i,d} = \x_d\}  +
	\frac{i - 1}{i} \widehat R_{i-1}(t,\x)  , \quad \text{with} \quad \widehat R_0(\cdot,\cdot)= 0.
\end{equation}
Here, $ g(i)>0$  is a  bandwidth which decreases with $i$ at a suitable rate. A same multivariate product kernel $\boldsymbol  K(\cdot) $ can be used for $\widehat{f}_{i}
(t, 1, \x  )$ and $\widehat R_i(t,\x) $. Moreover, for simplicity, we will set $h(i)=g(i)$ and use the same bandwidth for smoothing all the components of $\X_{i,c}$. In practice, the continuous covariates can be simply standardized beforehand. 

Based on \eqref{eq:haz_2}, we define the \emph{recursive conditional hazard (RCH)} estimator as
\begin{equation}\label{eq_rec_haza}
	\widehat \lambda _{T\mid \X, n} (t \mid \x) = \frac{\widehat{f}_{n}(t,1,\x  )}{ \widehat R_n(t,\x)  } , \qquad \x = (\x_c,\x_d) , \;  t>0.
\end{equation}
(The rule $0/0=0$ applies.) The conditional survival function estimator is then
\begin{equation}\label{eq_rec_surv}
	\widehat S _{T \mid \X,n} (t|\x) = \exp\left( - \int_0^t \widehat \lambda _{T \mid \X,n} (s\mid \x) ds \right).
\end{equation}

%\medskip

%\medskip

\begin{remark}\label{rem2v}
It is worth noting that $\widehat{f}_{n}(t,1,\x  )$ and $\widehat R_n(t,\x) $ can be rewritten under the form
\begin{equation}\label{sum_write}
\widehat{f}_{n}(t,1,\x  ) =  \frac{1}{n}\sum_{i=1}^n \delta_i U_{f,i}(t,\x  ) \qquad  and \qquad 
\widehat{R}_{n}(t,\x  ) =  \frac{1}{n}\sum_{i=1}^n U_{R,i} 	(t,\x  ),
\end{equation}
with
\begin{equation}\label{U_f}
	U_{f,i}(t,\x  ) =   L_{b(i)}\left(Y_i - t\right) \boldsymbol K_{h(i)}\left(\X_{i,c} - \x_c\right) \1 \{  \X_{i,d} = \x_d\} ,
\end{equation}
and
\begin{equation}\label{U_R}
	 U_{R,i} 	(t,\x  ) = \1 \{Y_{i}\geq t \}   \boldsymbol  K_{g(i)} \left(  \X_{i,c} - \x_c\right)\1 \{  \X_{i,d} = \x_d\} .
\end{equation}
These alternative expressions are more convenient for the theoretical study. 
\end{remark}

%
%\begin{remark}
%The estimator $\widehat{f}_{i}(t, 1, \x  )$ is only updated with uncensored observations. 
%Let us consider the following recursion~: with $\widetilde{f} _{0}(\cdot,1,\cdot)=0$ and the rule $0/0=0$, for $ i\geq 1$, define 
%\begin{equation}\label{eq:fRNWuni2}
%	\widetilde{f}_{i} (t, 1, \x) = \frac{\delta_i U_{f,i}(t,\x  ) }{\delta_1+\cdots+\delta_i} 
%	+ \frac{\delta_1+\cdots+\delta_{i-1}}{\delta_1+\cdots+\delta_i} \widetilde {f} _{i-1}(t,1,\x) = \frac{1}{\delta_1+\cdots+\delta_i} \sum_{k=1}^i \delta_k U_{f,k}(t,\x  ).
%\end{equation}
%The identity $\widehat{f}_{n} (t, 1, \x) = \widetilde{f}_{n} (t, 1, \x)  \{\delta_1+\cdots+\delta_n\}/n$ then holds, and we deduce that $\widehat{f}_{n}  $ is also equal to a recursive estimator of $ f_{Y,\X\mid \delta}(t,\x\mid \delta=1)$ multiplied by the estimator of $\PP(\delta=1)$.
%% See also the second equality in \eqref{eq:haz_2}. 
%\end{remark}
%
%

%\medskip

\begin{remark}\label{boundary_t}
Given that the random variable $Y$ is strictly positive, it may give rise to problems in the estimation of the left endpoint of the support. Standard remedies for the boundary problem can be applied, see \cite{S1986}. For example, if the conditional density of $Y$ given $\X$ has a zero limit when $t\downarrow 0$, we can use a correction via reflection by redefining 
$$
	U_{f,i}(t,\x  ) =   \{L_{b(i)}\left(Y_i - t\right) - L_{b(i)}\left(Y_i + t\right)\}\boldsymbol K_{h(i)}\left(\X_{i,c} - \x_c\right) \1 \{  \X_{i,d} = \x_d\} , 
$$
and modifying $	\widehat{f}_{i} (t, 1, \x)$ accordingly. See \cite{S1986} for other simple corrections designed for other types of boundary behavior assumptions for the density.\end{remark}

%--------------------------------------------------------------
\section{Extending the scope}\label{sec:chap1-model-discussion}
%--------------------------------------------------------------
% !TeX root = ../hazard.tex

It is worth noticing that the idea we propose in this paper, that is to use the hazard representation \eqref{eq:haz_2} with recursive estimators for the denominator and the numerator, is general. First, there are many other recursive estimators that can be used. Second, our approach applies to other types of incomplete observation mechanisms for the duration $T$ than the random right censoring.   Details of these two types of extension are given below. 

\subsection{Alternative recursive schemes}\label{subsec_alt_rec}

With the notation from Remark \ref{rem2v}, a  general class of estimators for $f_{Y,\delta,\X}(t,1,\x )$ can be 
\begin{equation}\label{eq:alt_rec}
\widehat{f}_{n}(t,1,\x  ) =  \frac{1}{\Omega_n}\sum_{i=1}^n w_i \delta_i U_{f,i}(t,\x  ) ,\quad \text{ where } \; w_i>0 \; \text{ and }\; \Omega_n = \sum_{i=1}^n w_i
\rightarrow \infty.
\end{equation}
As recalled by \cite{MPS2009}, different choices of the weights $w_i$ lead to different types of recursive density estimators~: the choice $w_i=1$ leads to our estimator $\widehat f_i$ which is basically the one of \cite{WW1969} and \cite{Y1971}; if $b(i) =h(i) \sim i^{-\alpha}$ and $\boldsymbol K$ is the product kernel constructed with the univariate kernel $L$, setting $w_i=h(i)^{(d_c+1)/2}$  leads to a version of the estimator studied by \cite{WD1979}, while the choice  $w_i=h(i)^{d_c+1}$ yields the estimator studied by \citep[Ch. 7]{D2013}. 

For alternative local constant estimators for  $R(t,\x)$ see, for example, \cite{R1992}. 
%\citep[Section 5]{R1992}. 
A recursive version of the local linear estimator is studied by \cite{KX2019}.

\subsection{Conditional hazard representations beyond the right censoring}\label{sec_ext}

In the following, we establish the representation of the hazard function for several types of incompleteness mechanisms. While some of the following facts are more or less explicitly stated in the literature, several of the representations derived below appear to be new. 

\subsubsection{Left censoring}

A possible extension of our approach is to consider left-censoring instead of right-censoring. See, for example, \cite{GUO1991}. In the left-censoring context, instead of duration of interest $T$, the observed variable is $Y=T \vee C$ together with the indicator $\delta = \1\{T \geq C\}$. In this case, $C$ is the left-censoring. 
By Assumption \ref{id_ass}, the following relationships hold true~: 
\begin{equation}\label{eq:sys2}
	\left\{\begin{array}{ll}
		H_1(dt\mid \x) = F_{C\mid \X}(t\mid \x) F_{T\mid \X}(dt\mid \x) \\
		H_0(dt\mid \x) = F_{T\mid \X}(t-\mid \x) F_{C\mid \X}(dt\mid \x)
	\end{array}\right. ,
\end{equation}
with the $H_k$ defined in \eqref{eq:def_Hk}. 
%Since by definition and Assumption \ref{id_ass},
%$
%H(t\mid \x) = H_0(t\mid \x) + H_1(t\mid \x) = F_{T\mid \X}(t\mid \x)F_{C\mid \X}(t\mid \x),
%$
The system \eqref{eq:sys2} can be solved to deduce
\begin{equation}\label{reverse_l}
	r_{T\mid \X}(t\mid \x) := \frac{f_{T\mid \X}(t\mid \x)}{F_{T\mid \X}(t\mid \x)} =  \frac{f_{Y,\delta, \X}(t,1,\x)}{H(t\mid \x)f_\X(\x)}. 
\end{equation}
The function $r_{T\mid \X}$ is a conditional reverse time hazard function which, like the standard hazard function, characterizes the conditional distribution of $T$ given $\X=\x$ because the relationship  $F_{T\mid \X}(t\mid \x)=\exp(-\int_t^\infty r_{T\mid \X}(s\mid \x) ds)$ holds true. 
A recursive estimator of $r_{T\mid \X}(t\mid \x)$ is  obtained from $\widehat{f}_{n}(t,1,\x  ) $ in \eqref{sum_write} and the estimator of $H(t\mid \x)f_{\X}(\x) $ defined as 
\begin{equation}\label{U_f2}
	\frac{1}{n}\sum_{i=1}^n V_{L,i} 	(t,\x  ), \qquad \text{ with} \qquad 	V_{L,i} 	(t,\x  ) = \1 \{Y_{i}\leq t \}   \boldsymbol  K_{g(i)} \left(  \X_{i,c} - \x_c\right)\1 \{  \X_{i,d} = \x_d\} .
\end{equation}

\subsubsection{Right censoring with a cure proportion}

In many applications a fraction of the right censored observed durations corresponds to subjects who will never experience the event of interest. In biostatistics such models are usually called \emph{cure models}. Economists sometimes call such models \emph{split population models}, while the reliability engineers refer to them as \emph{limited-failure population} life models. The usual approach for such models is to suppose that the support of the duration $T$ is $[0,\infty]$, while the right censoring $C$ has no mass at infinity. See \cite{PVK2020} for a discussion. Let $\phi(\x)=\PP(T<\infty \mid \x)>0$ be  the conditional probability of being susceptible to the event of interest given the covariates, and $1-\phi(\x)$ be the conditional probability of being cured. The second equation in \eqref{eq:sys1} then becomes 
\begin{equation}\label{eq:sys1c}
		H_0(dt\mid \x) = S_{T\mid \X}(t\mid \x) F_{C\mid \X}(dt\mid \x) 
=\{1-\phi(\x)+\phi(\x) S_{0,T\mid \X}(t\mid \x)\} F_{C\mid \X}(dt\mid \x),
\end{equation}
where $S_{0,T\mid \X}(t\mid \x)=\PP(T>t\mid T<\infty , \X=\x)$ is the (proper)  survival function of the susceptibles (often called the latency of the model).
For identification purposes it is commonly assumed that, for any $\x$,  $S_{0,T\mid \X}(\cdot\mid \x)$ has a bounded support, say, $[0,\tau(\x)]$ and $\inf_{t\in [0,\tau(\x)]}S_{C\mid \X}(t\mid \x) \geq \mathfrak c>0$, for some constant $\mathfrak c$ (not depending on the covariates). By  \cite[Equation (2.10)]{PVK2020}, we get the following expression for the conditional hazard function of the susceptibles~:
\begin{equation}\label{eq:haz_17}
	\lambda_{0,T\mid \X}(t\mid \x) 
	%=  \frac{f_{Y,\delta,\X}(t,1,\x)}{\{1 - H(t-\mid \x)\}f_\X(\x)- \{1-\phi(\x)\}S_{C\mid \X}(t-\mid \x)f_\X(\x)}\\
	= 
\frac{f_{Y,\delta,\X}(t,1,\x)}{R(t,\x)- \{1-\phi(\x)\}S_{C\mid \X}(t-\mid \x)f_\X(\x)}	, \qquad t\in[0,\tau(\x)],
\end{equation}
with $R(t,\x)$ defined in \eqref{eq:haz_2}. It is worth noticing that $\{1-\phi(\x)\}$, $S_{C\mid \X}(t-\mid \x)$ and $f_\X(\x)$ can also be estimated recursively. Indeed, by definition $1-\phi(\x)= S_{T\mid \X}(\tau(\x)\mid \x)$ as defined in \eqref{eq:S_1}, and therefore the conditional probability of being cured can simply be estimated by the recursive survival estimator in \eqref{eq_rec_surv} considered with $t=\tau(\x)$. Next, by switching the roles of the variables $T$ and $C$ and redefining  the censoring indicator as $1-\delta$, $C$ is right-censored by $T$. Then, since $C$ has no mass at infinity, assuming it admits a conditional hazard, an estimator like in \eqref{eq_rec_surv} can be constructed to estimate  $S_{C\mid \X}(t\mid \x)$. Finally, $f_\X(\x)$ can be estimated by a standard recursive density estimator, that is a simplified version of $\widehat f_n$ in \eqref{sum_write}.

\subsubsection{Left truncation and right censoring}

In addition to random right censoring, we can also consider a left truncation mechanism, which also occurs naturally in the applications. 
For example in a medical study  conducted over the given period about the longevity of patients after a surgery, the data are left truncated because only patients who were alive at the beginning of the study may be included in the study. The survival function estimator with left truncation, analogous to the  Kaplan-Meier estimator, is the so-called Lynden-Bell estimator; see \cite{LB1971}, \cite{W1985}. To introduce the LTRC mechanism, that is \emph{left-truncation in the presence of right-censoring}, we follow the lines in \citep[Section 6.4]{E2018}. There is hidden (latent) vector $(T^*,C^*,\X^*, L^*)$, where $L^*$ is the truncation variable, $T^*$ is the variable of interest, $C^*$ is the right censoring, and $\X^*$ are the predictors. Let  $(T_l^*,C_l^*,\X_l^*, L^*_l)$, $l\geq1$, be independent copies  of $(T^*,C^*,\X^*,L^*)$. If $L^*_l > Y^*_l:=T_l^*\wedge C_l^*$, the $l-$th realization in the latent sample is not observed, while if $L^*_l \leq  Y_l^*$ the realization $(Y_l^*,\delta_l^*,\X_l^*, L^*_l)$, with $\delta_l^*=\1\{T_l^*\leq C_l^*\}$,  becomes the $i-$th observation  $(Y_i,\delta_i, \X_i, L_i)$ of the observed variables $(Y,\delta, \X, L)$. Assuming the mutual conditional independence 
$
T^* \perp L^* \perp  C^* \mid \X^*,
$
and that  $\PP(L^* \leq  T^*\wedge C^*\mid \X^*=\x) \geq c>0$, for some constant $c$, 
it can be shown that 
\begin{equation}\label{eq:haz_11}
	\lambda_{T^*\mid \X^*}(t\mid \x) = \frac{f_{T^*\mid \X^*}(t\mid \x)}{S_{T^*\mid \X^*}(t\mid \x)} =  \frac{f_{Y,\delta ,\X}(t,1,  \x)}{\mathbb P (L\leq t \leq Y \mid \X=\x)f_{\X}(\x)},\qquad \forall t\geq 0,\; \x\in\mathcal X.
\end{equation}
See \citep[Eq. (6.4.9)]{E2018}. Here, $\lambda_{T^*\mid \X^*}$, $S_{T^*\mid \X^*}$ and $f_{T^*\mid \X^*}$ are the conditional hazard, the conditional survival and the conditional density of $T^*$ given $\X^*$, respectively. The numerator in \eqref{eq:haz_11} can be recursively estimated by $\widehat f_n$ in \eqref{sum_write}, while the recursive estimator of the denominator can be defined as $\widehat R_n$ but with $\1 \{Y_i\geq t\}$ replaced by $\1\{L_i\leq t \leq Y_i\}$ in  \eqref{U_R}. If  $\PP(L^* \leq  T^*\wedge C^*\mid \X=\x) =1$, the representations \eqref{eq:haz_1} and \eqref{eq:haz_11} coincide.

It is worth noting that both the truncation mechanism and the right censoring \emph{with} a cure rate can be present. A representation of the conditional hazard function of the susceptibles (the individuals with $ T^*<\infty$)
can be derived by combining \eqref{eq:haz_17} and \eqref{eq:haz_11}, that is 
\begin{equation}\label{eq:haz_108}
	\lambda_{0,T^*\mid \X^*}(t\mid \x) =  \frac{f_{Y,\delta,\X}(t,1,\x)}{\mathbb P (L\leq t \leq Y \mid \X=\x)f_\X(\x)- \{1-\phi(\x)\}\mathbb P (L\leq t \leq C \mid \X=\x)f_\X(\x)}	, 
\end{equation}
$t\in[0,\tau(\x)]$. Details are provided in the Appendix. For  the cure rate function $1-\phi(\x)$ we can use the identity $1-\phi(\x)= \PP(T^*=\infty\mid \x)= S_{T^*\mid \X^*}(\tau(\x)\mid \x)$ and express $S_{T^*\mid \X^*}(\cdot\mid \x)$ using  $\lambda_{T^*\mid \X^*}$ in \eqref{eq:haz_11}. The conditional probabilities  $\mathbb P (L\leq t \leq Y \mid \X=\x)$ and $\mathbb P (L\leq t \leq C \mid \X=\x)$ have representations as $R(t, \x)$  in \eqref{eq:haz_2}, with   $\1 \{Y\geq t \}$  replaced by  $\1 \{L \leq t \leq Y\}$ and $\1 \{L \leq t \leq C\}$, respectively. The identity \eqref{eq:haz_108} requires $C$ to be available for each observation. This is often the case  when the censoring value is determined by the follow-up period; see \cite{F2015} and \cite{CFK2015} for some recent examples in economics.

\subsubsection{A modified current status model}

As another example where our approach applies, let us consider the case studied by \cite{PR2006}, where instead of the lifetime of interest $T,$ one observes independent copies of a finite nonnegative duration $Y$ and of a discrete variable $A\in\{0,1,2\}$   such that
\begin{equation}\label{datai}
Y<T \; \text{ if } \;  A=0, \quad Y=T \; \text{ if } \;  A=1 \quad \text{ and} \quad  Y\geq T \; \text{ if } \;  A=2.
\end{equation}
A classic example of this type of censoring, which can be  called \emph{modified current status}, comes from the Stanford-Palo Alto Peer Counseling
Program; see \cite{HKJ1975}. In this study, 193 high schools boys were asked ‘When did you first use marijuana?’. The answers were the exact ages ($A=1$), ‘I have used it but cannot recall just when the first time was’ ($A=2$) and ‘I never used it’ ($A=0$). The variable of interest $T$ is the age at first use of marijuana. See also \cite{CRFM1991} for another example of application. It is worth noting that in \eqref{datai} the limit case where the event $\{A=2\}$ (resp. $\{A=0\}$) has zero probability corresponds to the usual random right-censoring (resp. left-censoring) setup, while the case where  $\PP(A=1)=0$ corresponds to the current status framework.

Extending the setup of \cite{PR2006} to regression, in addition to a non-negative censoring variable $C$, let $\Delta$ be a latent Bernoulli variable with conditional success probability $p(\x) \in (0,1]$. The observed variables $(Y,A,\X)$ are obtained as follows~: 
\begin{equation}\label{PR_Xx}
\left\{
\begin{array}{lll}
	(Y,A,\X)=(C,0,\X) & \mbox{if} &   C < T \\
	(Y,A,\X)=(T,1,\X) & \mbox{if} &    T\leq C\mbox{ and }\Delta=1\\
	(Y,A,\X)=(C,2,\X) & \mbox{if} &  T\leq C\mbox{ and }\Delta=0 \\
\end{array}
\right. .
\end{equation}
The conditional distribution  of  $(Y,A)$ given $\X=\x$ is given by the sub-distribution functions
\begin{equation}\label{Hk_rolin}
H_k(t \mid \x)=\PP(Y\leq t, A=k \mid  \X=\x),\quad t\geq 0, \;k\in\{0,1,2\},\; \x\in \mathcal X.
\end{equation}
Let $H(\cdot \mid \x)= H_0(\cdot \mid \x)+H_1(\cdot \mid \x)+H_2(\cdot \mid \x)$, and let $H(dt \mid \x)$, $H_k(dt \mid \x)$, denote the associated measures for $H(\cdot \mid \x)$, $H_k(\cdot \mid \x)$, $k\in\{0,1,2\}$, respectively. 

In addition to Assumption \ref{id_ass}, assume that 
\begin{equation}\label{id_ass_b}
	\Delta \perp (T,C) \mid \X.
\end{equation}
Then, with the notation from Section \ref{sec:chap1-model-methodo}, for any $\x$ in the support of $\X$, we have
\begin{equation}\label{model_eq1PR}
	\left\{
	\begin{array}{ccl}
		H_0(dt\mid\x)  & = & S_{T\mid \X}(t\mid\x) F_{C\mid \X}(dt\mid\x) \\
	 	H_1(dt\mid\x)  & = & p (\x) S_{C\mid \X}(t-\mid\x)F_{T\mid \X}(dt\mid\x)\\
		H_2(dt\mid\x)  & = & \{1-p(\x) \} F_{T\mid \X}(t\mid\x)F_{C\mid \X}(dt\mid\x)
	\end{array}
	\right. .
\end{equation}
Notice that by the condition \eqref{id_ass_b} and the definition of the model \eqref{PR_Xx},
\begin{equation}\label{inv_pPR}
	p(\x)=\frac{\mathbb{P}(\Delta = 1,  T\leq C\mid \X=\x)}{\mathbb{P}(T\leq C\mid \X=\x)} =\frac{\PP (A=1\mid \X=\x)f_{\X}(\x)}{\PP (A\in\{1,2\}\mid \X=\x)f_{\X}(\x)}.
\end{equation}
Next, after multiplying both sides of the first equation in \eqref{model_eq1PR} by $p(\x)$, then adding it to the second equation, and finally integrate over $[t,\infty)$, we get
$$
1-H_1(t-\mid \x)+p(\x)\{1-H_0(t-\mid \x) \}= p(\x) S_{T\mid \X}(t\mid \x)S_{C\mid \X}( t- \mid \x),\qquad t\geq 0.
$$ 
Finally, using this for dividing both sides in the second equation in \eqref{model_eq1PR}, we get the following  representation for the conditional hazard function of the duration of interest $T$~:
\begin{equation}\label{eq:haz_3}
	\lambda_{T\mid \X}(t\mid \x)  
	= \frac{f_{Y,A,\X}(t,1,\x )}   {[1-H_1(t-\mid \x)+p(\x)\{1-H_0(t-\mid \x) \}]f_{\X}(\x)} ,  
\end{equation}
where $f_{Y,A,\X}$ is the joint density of the observed variables. In particular, it holds   $H_1(dt\mid\x)f_{\X}(\x)=f_{Y,A,\X}(t,1,\x )dt$.
The values $f_{Y,A,\X}(t,1,\x )$  can be recursively estimated by 
\begin{equation}\label{sum_write3}
	\widehat{f}_{n}(t,1,\x  ) =  \frac{1}{n}\sum_{i=1}^n A_i(2-A_i) U_{f,i}(t,\x  ) ,
\end{equation}
with $U_{f,i}(t,\x  )$ defined in \eqref{U_f}. Meanwhile, the denominator and the numerator of $p(\x)$ in \eqref{inv_pPR}, and the $\{1-H_0(t-\mid \x)\}f_{\X}(\x)$, $k=0,1$, can be estimated by suitably replacing the indicators $ \1 \{Y_{i}\geq t \}  $  in the definition of $\widehat{R}_{n}(t,\x  ) $ in \eqref{sum_write} and \eqref{U_R}. The details are omitted. 

The incompleteness mechanism defined in \eqref{PR_Xx} is a combination of current status and right censoring. A representation like \eqref{eq:haz_3} can be similarly derived for the current status combined with left censoring, see \cite{PR2006}. 

\subsubsection{Competing risks}

In the competing risks context, the individuals can potentially experience more than one type of a certain event, but the occurrence of one type of event will prevent the occurrence of the others. Let $R_j$ be the time until the event of type $j$, with $1\leq j\leq J$, and let $T=\min\{R_j:1\leq j\leq J\}$. Following the notation of \cite{FG1999}, the observed variables are $(Y,\Delta, \Delta \varepsilon, \X)$ where $Y=\min(T,C)$, with $C$ the right censoring variable and $\Delta=\1\{T\leq C\}$, and $\varepsilon =j$ if  $T=R_j$. In the most realistic framework, the variables $R_j$ are not necessarily independent. In this case, the focus is on the cumulative incidence functions (CIF), see \cite{FG1999}, \cite{KK2012}. The conditional version of the CIF are  $F_{R_j\mid \X}(t\mid \x)=\PP(T\leq t , \varepsilon =j\mid \X=\x)$, $1\leq j \leq J$. 

Since the interest here lies in the CIF, we derive a convenient representation of the derivative of $F_{R_j\mid \X}(t\mid \x)$ with respect to $t$, similar to that for the conditional hazard discussed above. For this purpose,  it is worth noting that the setup of competing risks with right censoring can be equivalently defined using the $J$ individual times $R_j$, each of them censored by $C$. We consider the identification assumption $\{R_1,\ldots,R_J\}\perp C\mid \X$, which implies $T\perp  C \mid \X$. For some 
$1\leq j \leq J$ corresponding to the type of event of interest, let
$$
H_j(t\mid \x) =  \PP (Y\leq t,\Delta=1, \Delta\varepsilon = j\mid \X=\x), 
$$
and assume it admits the derivative $f_{Y,\Delta,\Delta\varepsilon\mid \X}(t,1,j\mid \x)$. Then the derivative (intensity)  $	f_{R_j\mid \X}$ of the conditional CIF of $R_j$ given $\X=\x$ can be written as 
\begin{equation}\label{cond_CIF}
	f_{R_j\mid \X}(t\mid \x) = S_{T\mid \X}(t\mid \x)\frac{f_{Y,\Delta,\Delta\varepsilon\mid \X}(t,1,j\mid \x)}{\EE(\1\{Y\geq t\}\mid \X=\x)},
\end{equation}
with the conditional survival $S_{T\mid \X}$ from \eqref{eq:S_1}. 
Note that $f_{Y,\Delta,\Delta\varepsilon\mid \X}(t,1,j\mid \x)f_{\X}(\x)$ can be estimated as in \eqref{eq:fRNWuni} by simply replacing $\delta_i$ with $\1\{\Delta_i\varepsilon_i=j\}$. Our representation \eqref{cond_CIF}, and the recursive kernel estimator that can be constructed using it, provide a model-free alternative to the proportional hazards approach of \cite{FG1999}.

\subsubsection{Bivariate time-to-event}

The approach can also be extended to bivariate durations under right censoring of each component. See \cite{D1988} and \cite{E2022}. In this case, there are three representations like \eqref{eq:haz_1} and \eqref{eq:haz_2} that can be established, and for each of them recursive estimates of the denominator and the numerator can be constructed. The details are omitted. 

\subsubsection{Open problems}

We have shown in this section that convenient conditional hazard representations depending on the distribution of the observed variables can be constructed in the presence of several other types of incompleteness. As soon as such representations can be derived, recursive estimators such as  proposed in Section \ref{sec:chap1-model-estimation} for the case of random right censoring can be constructed. Combining different types of incompleteness, the range of possible applications of our approach can be further extended. Such additional extensions are left for future study.

However, there are other situations where it is not clear whether there is a simple (conditional) hazard function representation based on the characteristics of the distribution of the observed variables. Interval censoring, see \cite{G1994}, and the left and right censoring model introduced by \cite{T1974}, are few examples. With interval censoring, the joint density in the numerator of the representation can no longer be identified from the data. In \cite{T1974}'s model there are two censoring variables $L\leq R$ and the observed variables are $Y =   \max[\min(T ,R),L] = \min[\max(T , L),R]$ and 
$A = 1$ if $L<T \leq R$ (no censoring), $A=0$ if $R<T$ (right censoring), and $A=2$ if $T\leq L$ (left censoring). 
Then, with the definition of $H_k$ as in \eqref{Hk_rolin} and obvious notation, the system \eqref{model_eq1PR} relating the conditional distribution of $(L,T,R)$ to the conditional distribution of the observed variables, becomes
\begin{equation}\label{model_Turn}
	\left\{
	\begin{array}{ccl}
		H_0(dt\mid\x)  & = & S_{T\mid \X}(t\mid\x) F_{R\mid \X}(dt\mid\x), \\
		H_1(dt\mid\x)  & = &  \{F_{L\mid \X}(t-\mid\x) - F_{R\mid \X}(t\mid\x) \} F_{T\mid \X}(dt\mid\x)\\
		H_2(dt\mid\x)  & = & F_{T\mid \X}(t\mid\x)F_{L\mid \X}(dt\mid\x)
	\end{array}
	\right. .
\end{equation}
This system does not admit an explicit solution leading to a representation similar to \eqref{eq:haz_3}.

\subsection{Product-integral representation of the survival function}

Given the conditional hazard, the conditional cumulative hazard $\Lambda_{T\mid \X}$ is obtained by integration and the conditional survival function follows by \eqref{eq:S_1}. It is worth noting that $\Lambda_{T\mid \X}$ has a simple integral representation even when the conditional distribution of the variable of interest $T$ does not admit a density. Indeed, from the facts presented above, we have
\begin{equation}\label{Lambda_rep}
	 \Lambda_{T\mid \X} (t\mid \x)= \int_{0}^t\frac{H_1(ds\mid \x)}{D(s\mid \x)},\qquad t\geq 0,
\end{equation}
where $H_1(Y\leq s \mid \x) = \PP (Y\leq s, Y=T \mid \X = \x )$, $H_1(ds\mid \x)$ is the associated measure,  and $D(s\mid \x)$ is some regression-like function given by the incompleteness mechanism.  Here, $Y$ is observed instead of $T$ due to an incompleteness mechanism as presented in Section \ref{sec_ext}. Let $I$ be an observed incompleteness indicator, taking the value 1 when $Y=T$. For example $I=\delta$ with left  or right censoring, and $I=A$ with the modified current status. 

The denominator  $D(s\mid \x)$ is a closed form function of the distribution of the observed variables. For example, $D(s\mid \x)=R(s, \x)/f_{\X}(\x)$ for the right censoring case, 
$$
D(s\mid \x) = \mathbb P (L\leq s \leq Y \mid \X=\x)- \{1-\phi(\x)\}\mathbb P (L\leq s \leq C \mid \X=\x),
$$
with the LTRC mechanism and a cure proportion if $C$ is always observed, including the case $\phi(\x)=1$ (no cure) when $C$ does not have to be always observed, and 
$$
D(s\mid \x) =  1-H_1(s-\mid \x)+p(\x)\{1-H_0(s-\mid \x) \},
$$
for the modified current status model. 

The conditional measure $H_1(ds\mid \x)$ does not require a density. However, most of the applications are covered by the assumption that $H_1(ds\mid \x)$ can be decomposed into an atomic measure, with $\mathcal D(t;\x)$ the set of atoms in $[0,t]$ given $\X=\x$, and an absolutely continuous measure $H^c _1(ds\mid \x)$. If $f^c_{Y,I\mid \X}(\cdot,1\mid \x)$ is the %Radon-Nikodym 
derivative of $H^c _1(\cdot\mid \x)$ given $\X=\x$, 
%that is $H^c _1(ds\mid \x)=f^c_{Y,I\mid \X}(s,1\mid \x) ds$, 
we get
\begin{multline}\label{Lmbda_c}
		 \Lambda_{T\mid \X} (t\mid \x)= \int_{0}^t\frac{f^c_{Y,I\mid \X}(s,1\mid \x)}{D(s\mid \x)}ds + \sum_{s\in \mathcal D(t;\x)} \frac{H_1(\{s\}\mid \x)}{D(s\mid \x)}\\=:\Lambda^c_{T\mid \X} (t\mid \x)+\sum_{s\in  \mathcal D(t;\x)}\Lambda_{T\mid \X} (\{s\}\mid \x) .
\end{multline}
By product-integration \citep{GJ1990}
the conditional survival can be written as 
\begin{equation}\label{prod_intX}
	S_{T\mid \X} (t\mid \x) = \Prodi_{s\in[0,t]}\left(1-\Lambda_{T\mid \X} (ds\mid \x)  \right)= \exp\{-\Lambda^c_{T\mid \X} (t\mid \x) \} \Prodi_{s\in  \mathcal D(t;\x)}\left(1-\Lambda_{T\mid \X} (\{s\}\mid \x)  \right).
\end{equation}
Recursive  estimators, as considered in Section \ref{sec:chap1-model-estimation}, can be used for $f^c_{Y,I\mid \X}(s,1\mid \x)$, $H_1(\{s\}\mid \x)$ and $D(s\mid \x)$, from which  estimators for $\Lambda_{T\mid \X} (t\mid \x)$ and $S_{T\mid \X} (t\mid \x)$ can be constructed. 
%In the case without covariates, the estimators of $H_1$ and $D$ can be constructed using empirical means. 

%--------------------------------------------------------------
\section{Asymptotic results}\label{sec:chap1-model-theory}
%--------------------------------------------------------------
% !TeX root = ../hazard.tex

For simplicity, the theoretical study focuses mainly on the case of right censoring. However, the theory for the estimators to be defined in the situations from Section \ref{sec:chap1-model-discussion} is similar up to simple changes. The proofs of the results in this section are given in the Appendix. 

The asymptotic properties of the RCH estimator $\widehat \lambda _{T\mid \X, n}$ in \eqref{eq_rec_haza} can be derived from those of $\widehat{f}_{n}$ and $\widehat R_n$. The asymptotics for  $\widehat S _{T\mid \X, n}$ follows using the relationship \eqref{eq:S_1}. 
For the asymptotic theory for $\widehat{f}_{n}$ and $\widehat R_n$ on compact intervals, we consider a generic recursive estimator 
\begin{equation}\label{generic}
\widehat{g}_{n}(t,\x) =  \sum_{i=1}^n \omega_{n,i} D_i(t) \mathbb K_{\mathfrak h(i)} (\Z_i-\z  )  \1 \{  \X_{i,d} = \x_d\}, \qquad t\in \mathcal T, \x_d\in \mathcal X_d,\z\in \mathcal Z,  
\end{equation}
where $\mathcal T \subset [0,\infty)$ is some compact interval in the support of $T$, $\z =(t,\x_c)$ or $\z=\x_c$,  and
\begin{itemize}
\vspace{-.2cm} 	\item $\Z_i$ are distributed as the continuous random vector $\Z\in\mathcal Z$ involved in the smoothing, that means $\Z = (Y,\X_c)$ and $\mathcal Z = \mathcal T \times \mathcal X_c\subset \mathbb R^{1+d_c}$, or $\Z=\X_c$ and $\mathcal Z =  \mathcal X_c\subset \mathbb R^{d_c}$;
\vspace{-.2cm} 	\item  $\X_{i,d}$ are distributed as $\X_d\in \mathcal X_d$ which gathers the discrete predictors;
\vspace{-.2cm} 	\item $D_i(\cdot)$ are  random functions distributed as  $D(t)\in [0,1]$,  $t\in \mathcal T$;
\vspace{-.2cm} 	\item $\mathbb K_{\mathfrak h}(\z ) = \mathfrak h^{-p} \mathbb K (\z/ \mathfrak h)$ with $\mathbb K (\cdot)$ a multivariate kernel, and  $\mathfrak h(i)>0$ are bandwidths;
\vspace{-.2cm} 	\item $\omega_{n,i} $, $1\leq i \leq n$, are positive, non random weights that sum to 1. 
\end{itemize}
\vspace{-.1cm} The estimator $\widehat{g}_{n}$ becomes $\widehat{f}_{n}$ from \eqref{sum_write} if $\Z=(Y,\X_{c})$, $D(t)\equiv \1\{T\leq C\}$,   $\mathbb K$ is the kernel product between $L$ and $\boldsymbol K$, and $\omega_{n,i}=1/n$. For simplicity,
we set $\mathfrak h(i)=b(i)=h(i)$ for $\widehat{f}_{n}$, which is reasonable if all variables have been rescaled beforehand. 
If  $D(t)=\1\{Y\geq t\}$,   $\Z=\X_{c}$, $\mathbb K=\boldsymbol K$ and $\mathfrak h(i)=h(i)$, $\omega_{n,i}=1/n$, we get $\widehat{g}_{n}$ equal to  $\widehat{R}_{n}$ in \eqref{sum_write}. 

%It is worth noting that  $\widehat{f}_{n}$ and $\widehat{R}_{n}$ require smoothing in the dimension $p=d_c+1$ and $p=d_c$, respectively. 

\medskip

\begin{assum}\label{assump1}
	\begin{enumerate}[I)]
\vspace{-.2cm} 	\item\label{xia19} The support $\mathcal X_d$ of the discrete predictors is finite, and the set $\mathcal Z$   is a bounded hyperrectangle in $\mathbb R^p$ with $p=d_c+1$ (for $\widehat f_n$) or $p=d_c$ (for $\widehat R_n$). 
		
\vspace{-.2cm} 	\item\label{sample} The observations $(D_i(\cdot), \Z_i,\X_{i,d})$, $1\leq i\leq n$, are i.i.d. distributed as  $(D(\cdot), \Z,\X_{d})$.

\vspace{-.2cm} 	\item\label{xia2}  For any $\x_d\in\mathcal X_d$, the joint density $f_{Y,\delta,\X}(t, 1, \x)$ admits  partial  derivatives of second order with respect to $(t,\x_c)$ that are uniformly continuous on  $\mathcal Z = \mathcal T \times \mathcal X_c$.

\vspace{-.2cm} 	\item\label{xia2b}  For any $t\in\mathcal T$, $\x_d\in\mathcal X_d$, $R(t,\x)$ admits partial  derivatives of second order with respect to $\x_c$ over $\mathcal Z = \mathcal X_c$ and their modulus of continuity is independent of $t$.

\vspace{-.2cm} 	\item\label{xia2c} The realizations of $D(\cdot)$ are non-increasing, 
%the map $(t,\x)\mapsto R(t,\x)$ is Hölder continuous 
and $\inf_{(t,\x)\in\mathcal T \times \mathcal X}R(t,\x)>0$.
		
\vspace{-.2cm} 	\item\label{ker1}   $\mathbb K(\cdot)$ is a product kernel obtained with a symmetric and Lipschitz continuous density $K(\cdot)$  supported on $[-1,1]$, and  $\mu_2(\mathbb K)   = \int_{-1}^1 u^2   K (u) du$. Moreover,  $\mathfrak h(i)=ci^{-\alpha}$, with $\alpha \in (0,1/p)$, $0<c\in[\underline c, \overline c]$. The weights are $\omega_{n,i}=i^{\beta}/\sum_{j=1}^n j^{\beta}$, with $0\leq \beta \leq \alpha p<1$.  
	\end{enumerate}
\end{assum}

The conditions in Assumption \ref{assump1} are standard. The finite support of the discrete covariates  is assumed for convenience. It could be relaxed at the expense of more complicated writings, but it does not affect the main findings below. The requirement that $D(\cdot)$ is non-increasing is imposed for technical convenience, as it allows a simple argument in the proof of uniform convergence. Relaxing this condition, for example by requiring only bounded variation, is straightforward. The bandwidth condition in Assumption \ref{assump1}-\ref{ker1} makes $\widehat{g}_{n}$ to depend on $c$, and the uniform convergence with respect to $c$ will allow for data-driven bandwidths with a given decrease rate $\alpha$. However, for the sake of simplicity, the dependence of $\widehat{g}_{n}$ on $c$ is omitted from the notation here. The choice of  $\omega_{n,i}$, replacing $w_i/W_n$ in \eqref{eq:alt_rec}, allows the theory to include some more general recursive estimators as discussed above. 
%We consider that Assumption \ref{assump1} holds for the numerator (resp. denominator) of the RCH estimator $\widehat \lambda _{T\mid \X, n}$ if it holds with $\Z=(Y,\X_c)$ (resp. $\Z=\X_c$), $\mathcal Z = \mathcal T \times \mathcal X_c$ (resp. $\mathcal Z =  \mathcal X_c$) and $p=d_c+1$ (resp. $p=d_c$), $\mathbb K$ the kernel product between $L$ and $\boldsymbol K$ (resp. $\mathbb K=\boldsymbol K$), and $D_i(\cdot)=\delta_i$ (resp. $D_i(t)=\1\{Y_i\geq t\}$).

We first study the uniform rate of convergence  of the stochastic part of our estimator $\widehat{g}_{n}$. 

%\medskip

\begin{proposition}\label{stoch1}
	If the conditions \ref{xia19}, \ref{sample}, \ref{xia2c} and \ref{ker1} in  Assumption \ref{assump1} hold true,   then
	$$
	\sup_{c\in[\underline c, \overline c]} \sup_{ t\in \mathcal T} \sup_{\x \in\mathcal X} \left| \widehat{g}_{n}(t,\x) - \EE \left[\widehat{g}_{n}(t,\x)\right]\right| =  O_\PP \left( n^{-(1-\alpha p)/2}\sqrt{\log n } \right).
	$$
\end{proposition}

%\medskip

\begin{remark}
For simplicity, we consider here the i.i.d. framework. The uniform convergence in Proposition \ref{stoch1} can be obtained once suitable concentration inequalities for sums of centered variables are available, in particular under suitable types of weak dependence for stationary series. For example, a concentration inequality under $\varphi-$mixing dependence is used by \cite{WL2004} to derive strong uniform convergence of recursive estimators similar to our $\widehat{g}_{n}$. Alternatively, the short-range dependence condition considered by \cite{LXW2013} and their Rosenthal and Nagaev-type inequalities can be used. 
\end{remark}

%\medskip

\begin{remark}
In the case where $\mathcal T$ includes the left endpoint of the support of $T$, a boundary correction of the kernel estimator   is generally necessary. Proposition \ref{stoch1} remains true for a wide choice of boundary corrections. For example, with the correction mentioned in  Remark \ref{boundary_t}, the estimator  $\widehat{g}_{n}(t,\x)$ can be split into two parts and each part will satisfy the conditions required to derive the rate of uniform convergence as in Proposition \ref{stoch1}.
\end{remark}

%\medskip

To study the bias of $\widehat g_n$, let us denote its limit by $G(t,\x)$, that means 
\begin{equation}\label{eq:def-GR}
	G(t,\x)= f_{Y,\delta,\X}(t, 1, \x) \qquad  \text{ or } \qquad   G(t,\x)=R(t, \x), 
\end{equation}
depending on the weather $\widehat g_n$ estimates the numerator or denominator in \eqref{eq_rec_haza}.
Since the rate of decrease of the bias is expected to deteriorate near the boundary, the next result is stated for the set of values $(t,\x)$ away from the boundary of the support of $(T,\X)$. For $\epsilon >0$, let $\mathcal X_{\epsilon} =\{(\x_c,\x_d)\in\mathcal X: \|\x_c-\x^\prime\|\geq \epsilon, \forall \x^\prime \in \mathbb R^{d_c}\setminus\mathcal X_c\} $
and  $\mathcal T_\epsilon = [\epsilon, \max \mathcal T]$.

%$\mathcal Z_\epsilon =\{\z\in\mathcal Z: \|\z-\z^\prime\|\geq \epsilon, \forall \z^\prime \in \mathbb R^p\setminus\mathcal Z\}$.

%\medskip

\begin{proposition}\label{bias1}
	If the conditions of  Assumption \ref{assump1} hold true, then, for any $\epsilon>0$,
	$$
		 	\sup_{c\in[\underline c, \overline c]}  \sup_{t\in\mathcal T_\epsilon } \sup_{\x\in  \mathcal X_\epsilon}  \left|  \EE \left[\widehat{g}_{n}(t,\x)\right] - G(t,\x) \right| = O\left( n^{-2\alpha}\right).
	$$
\end{proposition}

%\medskip

Propositions \ref{stoch1} and \ref{bias1}  applied to the numerator and denominator in  \eqref{eq_rec_haza} leads to the following.

%\medskip

\begin{corollary}\label{cor_lambda1}
If Assumption \ref{assump1} holds true for the numerator and the denominator of the conditional hazard in \eqref{eq:haz_2}, and $\inf_{t\in\mathcal T, \x\in\mathcal X} R(t,\x) >0$, then, for any $\epsilon >0$,
	$$
\sup_{c\in[\underline c, \overline c]}  \sup_{t\in\mathcal T_\epsilon } \sup_{\x\in  \mathcal X_\epsilon}\left|   	\widehat \lambda _{T\mid \X, n} (t \mid \x) - 	 \lambda _{T\mid \X} (t \mid \x) \right| =O_\PP \left(n^{-2\alpha}+ \frac{\sqrt{\log n }}{n^{(1-\alpha (d_c+1))/2}}\right) .
	$$	
The optimal exponent for the bandwidth is $\alpha = 1/(d_c+5)$, which gives the rate of uniform convergence $O_\PP\left(n^{-2/(d_c+5)}\sqrt{\log n} \right)$ for the conditional hazard function estimator.
\end{corollary}

%\medskip

\begin{remark}
The rate $n^{-2/(d_c+5)} $ can be written as  $n^{-r/(2r +1)} $, where here $r= 2/(d_c+1)$ is the effective smoothness for an isotropic multivariate function of $d_c+1$ continuous variables when the second order derivatives exist. Recall that isotropic  means that the function has the same regularity, here the number of derivatives, in all variables. In the anisotropic case, $r(d_c+1)$ is defined as the harmonic mean of the regularity of the components. The rate $n^{-r/(2r +1)} $ is  optimal for the pointwise and the mean integrated squared error, in both isotropic and anisotropic case. See \cite{E2024}. The $\sqrt{\log n }$ factor is expected when considering the uniform rates of convergence. It is worth noting that the general approach we propose can be extended to the anisotropic case, at the expense of different bandwidths for different components of the vector $(Y,\Z)$ and more complicated writings, provided the effective smoothness $r$ is given.
\end{remark}

%\medskip

\begin{remark}
The construction of our recursive estimator requires bandwidth choices for   $\widehat f_n$ and $\widehat R_n$. For simplicity, we considered a same bandwidth for both. Note that $\widehat f_n$ requires smoothing in a higher dimension than  $\widehat R_n$. If both functions are isotropic with the same $r$,   $\widehat R_n$ converges  faster than $\widehat f_n$. For the purpose of  theoretical results under weaker assumptions, and for improving the small sample performance of the RCH  estimator, it is worth considering different bandwidths for $\widehat f_n$ and $\widehat R_n$. This is suggested by \cite{CM2020}, but such a refinement is beyond the scope of our current work.  
\end{remark}

%\medskip

\begin{remark} 
In the case of the estimator $\widehat f_n$, by the definition of $\mathcal Z_\epsilon$, Proposition \ref{bias1} excludes a small interval from the left-hand tail of $T$. The uniformity over the neighborhood of the left endpoint of $\mathcal T$ (assumed for simplicity to be the origin), can be achieved at the limit. More precisely, for $\widehat f_n$, we have  $\mathcal Z_\epsilon = \mathcal T_\epsilon \times \mathcal X_{c,\epsilon}$ and consider, for simplicity,  that $\mathcal T_\epsilon=[\epsilon,\max \mathcal T]$. Then, taking $\epsilon >0$ and a decreasing sequence, say,    $ \epsilon_n=n^{-\alpha \nu} $ with some suitable $\nu >0$ allows preserving the uniform rate of the bias in Proposition \ref{bias1} over $\mathcal T_{\epsilon_n}\times \mathcal X_{c,\epsilon}$, and thus let the domain of $t$ in the Corollary \ref{cor_lambda1} reach the origin at the limit. The exponent $\nu$ has to satisfy $0\leq \nu< (\beta+1-2\alpha)/ (\beta+1-\alpha)$. We can notice that taking larger $\beta >0$ allows for larger $\nu$ and thus smaller $\epsilon_n$. In other words, in theory, the choice of the weight $\omega_{n,i}$ in the definition of the estimator \eqref{generic} can help to alleviate the bias on the boundary of $\mathcal T$. See the discussion in the proof of Proposition \ref{bias1} for details. Finally, it should be pointed out that making $\epsilon_n$ to decrease to zero must be accompanied by a remedy that accounts for the non-vanishing bias that characterizes the usual kernel density estimates, such as a reflection of the points with respect to the origin. See the discussion in Remark \ref{boundary_t}. \end{remark}

\medskip

Concerning  the pointwise bias of the RCH estimator, from the proof of Proposition \ref{bias1} and using the first part of Lemma \ref{serie_Riemann}, if $\mathfrak h(i) = ci^{-\alpha}$ with  $\alpha = 1/(d_c+5)$, it can be shown that 
\begin{equation}\label{bias_lam1}
\EE \big[\; \widehat \lambda _{T\mid \X, n} (t \mid \x)\big ]-  \lambda _{T\mid \X} (t \mid \x) \\ = c^2\frac{n^{-2\alpha}(1\!+\!\beta)}{2(1+\beta \!-\!2\alpha) }\frac{\mu_2(\mathbb K) \operatorname{Trace}(\mathcal H_f(t,\x))}{R(t,\x)} 
\{1+o(1)\},
\end{equation}
where $\mathcal H_f(t,\x)$ is the Hessian of $f_{Y,\delta,\X}(t,1,\x)$ with respect to $(t,\x_c)$. 
On the other hand,  with $p=d_c+1$, 
the pointwise variance of conditional hazard estimator can be shown to be   
\begin{equation}\label{var_lam1}
	\operatorname{Var}\left( \widehat \lambda _{T\mid \X, n} (t \mid \x) \right) = \frac{c_{\mathbb K}}{n^{1 -\alpha p}} 
	\frac{ f_{Y,\delta,\X}(t,1,\x)} {c^pR^2(t,\x)}  \frac{(1+\beta)^2}{1+\alpha p+2\beta} \{1+o(1)\}.
\end{equation}
The optimal constant $c$ for defining the bandwidth $\mathfrak h(i)$ is simply obtained by minimizing the leading terms in the squared bias plus the variance, that is
\begin{equation}\label{opt_ct_RCH}
	c_{\rm opt}(\beta;\alpha,t,\x)=\left[\frac{p(1+\beta-2\alpha)^2c_{\mathbb K}f_{Y,\delta,\X}(t,1,\x)}{(1+\alpha p +2\beta)\{\mu_2(\mathbb K)\operatorname{Trace}(\mathcal H_f(t,\x))\}^2}\right]^{\frac{1}{p+4}}, \quad p=d_c+1, \alpha= 1/(p+4).
\end{equation}
Thus, the bandwidth minimizing the MSE of our RCH estimator is equal to the optimal bandwidth for the standard, non-recursive kernel estimator of $f_{Y,\delta,\X}(t,1,\x)$ multiplied by  
\begin{equation}\label{Cr_const}
c_r(\beta)  =   \left[\frac{\{d_c(1+\beta)+5\beta +3\}^2}{(d_c+5)\{2d_c(1+\beta)+10\beta+6\}}\right]^{1/(d_c+5)}<1.
\end{equation}
See \cite{KX2019} for the case $\beta=0$. A constant $c_r(\beta)<1$   is expected because a larger bandwidth used in the early recursions results in a large bias, which cannot be compensated for by the reduction in the variance. We deduce from above  that, if the estimation error for $R(t,\x)$ is negligible compared to that for $f_{Y,\delta,\X}(t,1,\x)$, the bandwidth rule for the RCH estimator can be derived from common rules for standard density kernel estimates.

%--------------------------------------------------------------
\section{Empirical analysis}\label{sec:chap1-emp-analysis}
%--------------------------------------------------------------

% !TeX root = ../hazard.tex

%2982

We first present simulation results for our recursive conditional hazard  estimator in the presence of  right censoring. The setup for the simulations  was designed to match the Rotterdam primary breast cancer dataset containing 2982 patients, as considered by \cite{RA2013}; see also \cite{FPLPSKBNJGHS2000}.  The event of interest is the occurrence of death from breast cancer after a surgery. Next, in Section \ref{sec:chap1-real-data} we analyze this real dataset.

\subsection{Simulation design}\label{sec:simu-implementation}

We consider two models for generating the time-to-event $T$: `Model 1' based on Cox's Proportional Hazard (CPH) model, \cite{C1972};  and `Model 2' using the Accelerated Failure Time (AFT) model, \cite{W1992}. Thus, the conditional hazards  given $\X=\x$ are defined as
$$
\lambda_{T\mid \X}(t\mid \x) = \lambda_0(t) \exp(\boldsymbol{x}^\top \boldsymbol{\beta})  \quad \text{ and }  \quad 
\lambda_{T\mid \X}(t\mid \x)=\lambda_0(t \exp(\x^\top \bbeta)) \exp(\x^\top \bbeta) ,
$$
respectively. Here, $\boldsymbol{x}^\top \boldsymbol{\beta}$ is the scalar product between $\boldsymbol{x}$ and $\boldsymbol{\beta}$. The  baseline hazard  is 
$$
\lambda_0(t) = \alpha_1 + \alpha_2 t + \alpha_3 t^2 + \alpha_4 \exp(\alpha_5 t).
$$
To generate a time-to-event $T$ starting from the conditional hazard function given $\X=\x$, we first compute the conditional survival function as the exponential of the minus conditional cumulative hazard function, and next use the inverse transform sampling method. 

The model coefficients $\boldsymbol{\beta}$ as well as the vector $\boldsymbol \alpha =( \alpha_1 ,\ldots,  \alpha_5)^\top $ are set to match the estimates obtained with the CPH when fitted to the breast cancer dataset. Details on the matching procedure are provided in the Supplement. The covariate vector $\X$, inspired by the  predictors available in the breast cancer dataset, is defined as follows. One continuous variable, named $\X_{c,1}$, is  supported on $[-3,3]$ with the distribution as the variable $3A-6$ where $A$ is generated from a   Beta distributions mixture  $\pi \mathcal B(a_1,b_1)+(1-\pi)\mathcal B(a_2,b_2)$, with parameters $a_1=17, b_1=10$,  $a_2=9, b_2=14$ and $\pi=0.4$.
A second continuous variable  $\X_{c,2}$ has a Gamma distribution   $\Gamma(0.38, 0.14)$, which means the expectation is approximately 2.709. The vector $\X$ is completed by three binary covariates. They are obtained from a multinomial variable with parameters $(0.47, 0.1, 0.43)$, with the first outcome accounting for a small tumor, and  a Bernoulli variable with parameter $0.51$ corresponding to the probability of the presence of a relapse in the breast cancer dateset. 
For Model 1 we set, 
$
\bbeta_{\rm CPH} = (0.24, 0.04, -0.69, -0.32, 1.9)^\top$   and    $\boldsymbol \alpha = (-20.32, -1.51, -0.08, 20.32, 0.08)^\top $,
while for Model 2 we keep the same $\boldsymbol \alpha$ but modify $\bbeta_{\rm AFT} = \bbeta_{\rm CPH}/4$ to make the hazard functions for the two models to have comparable scales. Finally, the censoring times $C$ are drawn from an exponential law of parameter $0.45$ and shifted by a suitable constant $\mu_C$ to obtain the  censoring proportions 20\%, 40\%, and 60\% for each type of simulation model.

\subsection{Simulation results}

Below we compare our RCH estimator of the conditional hazard  with that of
the semi-parametric CPH model fitted using the \texttt{R} package \texttt{survival}.  Given the estimate of $\bbeta$ by maximization of Cox's partial likelihood, the CPH estimate of the conditional hazard function is derived from Breslow's estimate (see \cite{B1972}) by finite differences, which are eventually smoothed by splines. Since Model 1 is a  proportional hazard model, we use the CPH estimator as benchmark for both Model 1 and 2.

For our method, we consider a product kernel constructed with the standard Gaussian density. Before estimating the joint density $f(t,1,\x)$ and the denominator $R(t,\x)$, all the continuous variables are standardized using the empirical standard deviations. The weights are set $\omega_{n,i}=1/n$.
For simplicity, the bandwidths $h(i)$ and $b(i)$ of the  estimator $\widehat{f}_{i} (t, 1, \x)$ and $g(i)$ of the estimator $\widehat R_i(t,\x)$ are considered of the form $c i^{-\alpha}$  with $\alpha = 1/(4+p)$, where $p=d_c+1$ and $p=d_c$ for $\widehat{f}_{i}$ and $\widehat{R}_{i}$, respectively. 
For the constant $c$ we adopt a simple rule-of-thumb. For  $\widehat{f}_{i} (t, 1, \x)$, we take it equal to the factor $c_r(0)$ in \eqref{Cr_const} multiplied by $\PP(\delta=1)^{1/(d_c+5)}$, which leads us to set the values of $c$ approximately equal to 0.656, 0.753 and 0.948 when the censoring rate is 20\%, 40\%, and 60\%, respectively. This reduction of $c_r(0)$ is justified by the fact that only the uncensored observations are involved in $\widehat{f}_{n}$. For the denominator, which has a faster rate of convergence, we consider $g(i)$ with $c=0.875$ for all censoring levels. Moreover, to avoid division by small values, we add $n^{-1}$ to the estimator $\widehat R_n(t,\x) $ in \eqref{eq_rec_haza}. Finally, we use a left boundary correction by reflexion in the numerator estimator, as discussed in Remark \ref{boundary_t}.  We consider sample sizes $n \in \{200, 400, 1000, 2000, 5000, 10000, 20000\}$ and each experiment is replicated $500$ times. 
The methods are evaluated on a  time grid $\mathcal G_T$ of 100 equidistant points in $ [0.1, 8]$. Regarding the covariate vector, we consider a grid $\mathcal G_{\X}$ of points by considering 25th, 50th, and 75th quantiles of the continuous covariates and the support of the discrete variables, which means 54 points. Finally, the global estimation error is evaluated by  
\begin{equation}\label{mise_crit}
	MISE = \frac{1}{100\times 54 } \sum_{t \in \mathcal G_T} \sum_{\x \in \mathcal G_\X} \left(\lambda(t|\x) - \widehat{\lambda}(t|\x) \right)^2.
\end{equation}

\begin{comment}

\begin{table}[h]
	
	\centering
	
	\caption{\small Median of the MISE over $500$ replications for different models, censoring rates, and sample sizes. } 
\small 			
	\begin{tabular}{lccccrrrrr}

		\hline
		
		model & censoring & estimator & 200 & 400 & 1000 & 2000 & 5000 & 10000 & 20000\\ 
		
		\hline
		
		Model 1 & 20\% & RCH & 0.0176 & 0.0127 & 0.0078 & 0.0053 & 0.0033 & 0.0023 & 0.0017 \\ 
		
		&   & CPH & 0.0024 & 0.0012 & 0.0005 & 0.0003 & 0.0001 & 0.0001 & 0.0001 \\ 
		
		& 40\% & RCH & 0.0156 & 0.0112 & 0.0067 & 0.0045 & 0.0027 & 0.0019 & 0.0014 \\ 
		
		&   & CPH & 0.0027 & 0.0013 & 0.0005 & 0.0003 & 0.0001 & 0.0001 & 0.0001 \\ 
		
		& 60\% & RCH & 0.0115 & 0.0080 & 0.0048 & 0.0032 & 0.0020 & 0.0014 & 0.0011 \\ 
		
		&   & CPH & 0.0033 & 0.0016 & 0.0007 & 0.0003 & 0.0001 & 0.0001 & 0.0001 \\ 
		
		\hline
		
		Model 2 & 20\% & RCH & 0.0501 & 0.0402 & 0.0253 & 0.0145 & 0.0063 & 0.0044 & 0.0042 \\ 
		
		&   & CPH & 0.0130 & 0.0111 & 0.0104 & 0.0101 & 0.0099 & 0.0098 & 0.0097 \\ 
		
		& 40\% & RCH & 0.0286 & 0.0236 & 0.0154 & 0.0096 & 0.0041 & 0.0027 & 0.0027 \\ 
		
		&   & CPH & 0.0104 & 0.0095 & 0.0082 & 0.0077 & 0.0076 & 0.0075 & 0.0076 \\ 
		
		& 60\% & RCH & 0.0165 & 0.0144 & 0.0108 & 0.0078 & 0.0041 & 0.0023 & 0.0015 \\ 
		
		&   & CPH & 0.0083 & 0.0069 & 0.0063 & 0.0060 & 0.0060 & 0.0059 & 0.0058 \\ 
		
		\hline
		
	\end{tabular}
	
	\label{tab:simu-mse}
	
\end{table}

\end{comment}

%Table \ref{tab:simu-mse} and 
Figure \ref{fig:re-mse-main} shows the relative error on the logarithmic scale when the data are generated in Model 1 and Model 2. As expected, the CPH estimator performs better in Model 1 which is a Cox's PH setup. The log-ratios deteriorate as the sample size increases because the CPH estimator has a parametric rate of convergence. In the AFT setup our nonparametric approach remains consistent and eventually outperforms the CPH, which is not consistent.

\begin{figure}[ht]
    \centering
    \includegraphics[width=0.65\linewidth, height=8cm]{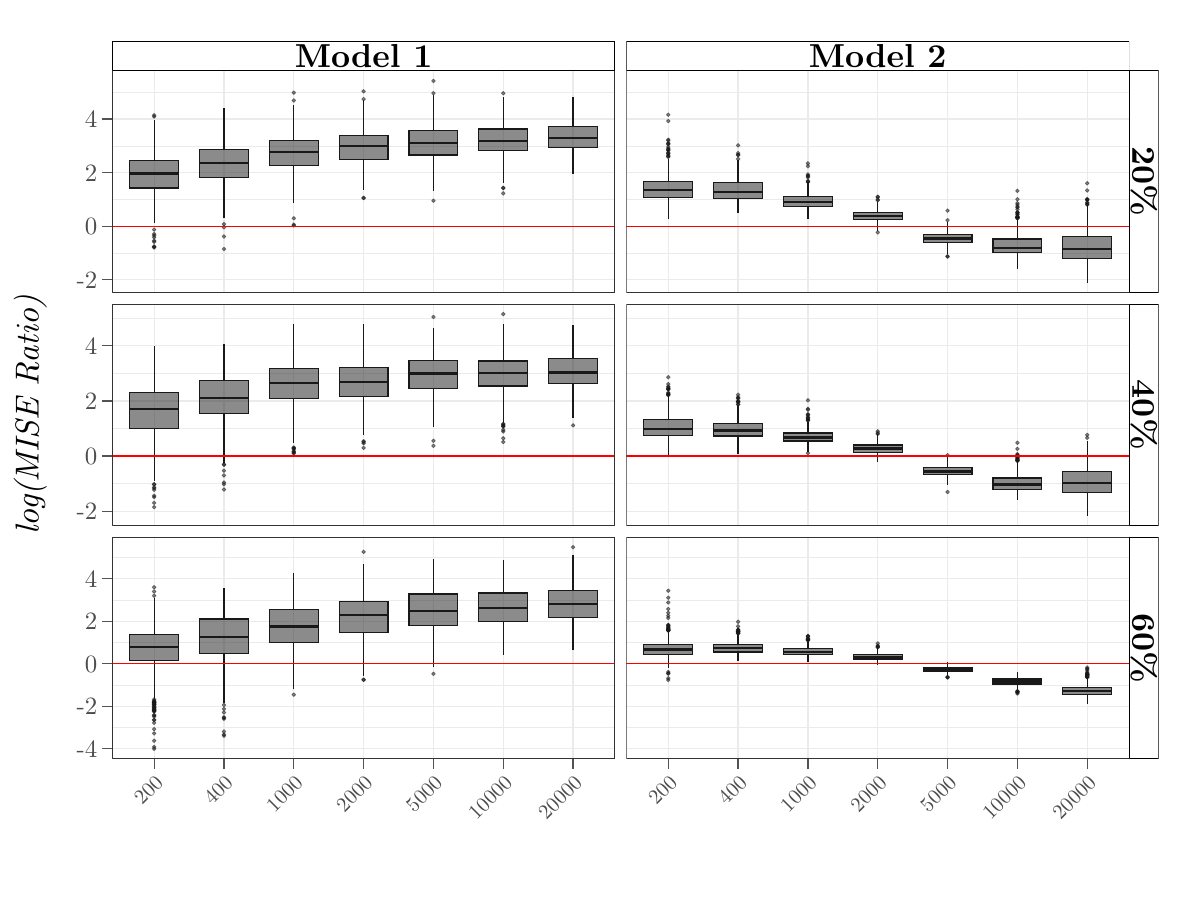}\vspace{-.6cm}
    \caption{\small The nonparametric RCH method compared with the conditional hazard function estimator obtained   in Cox's Proportional Hazard (CPH) model. The boxplots present the logarithm of the ratios of the MISE criterion \eqref{mise_crit} from $500$ replications, for different censoring rates in Model 1 and 2. Boxplots below the horizontal line indicate better performance for RCH.}
    \label{fig:re-mse-main}
\end{figure}

Figure \ref{fig:simu-comp-hazard} illustrates the differences between the estimation approaches for one covariate vector value, when the sample size is $n=2000$ and there is   20\% censoring. The conditional hazard is evaluated at the average value of the shifted Beta-mixture variable $\X_{c,1}$, the median value for the Gamma distributed $\X_{c,2}$, and the  small tumor modality for the multinomial variable. The top panels illustrate the scenario with no relapse, while the bottom panels consider the case with relapse. Our estimator satisfactorily recovers the shape of the conditional hazard function. As expected, the CPH estimator outperforms it in Model 1. However, RCH estimator performs better overall when the data are generated in the AFT model.

\begin{figure}[ht]
    \centering
    \vspace{-.2cm}
    \includegraphics[width=0.7\linewidth, height=7.8cm]{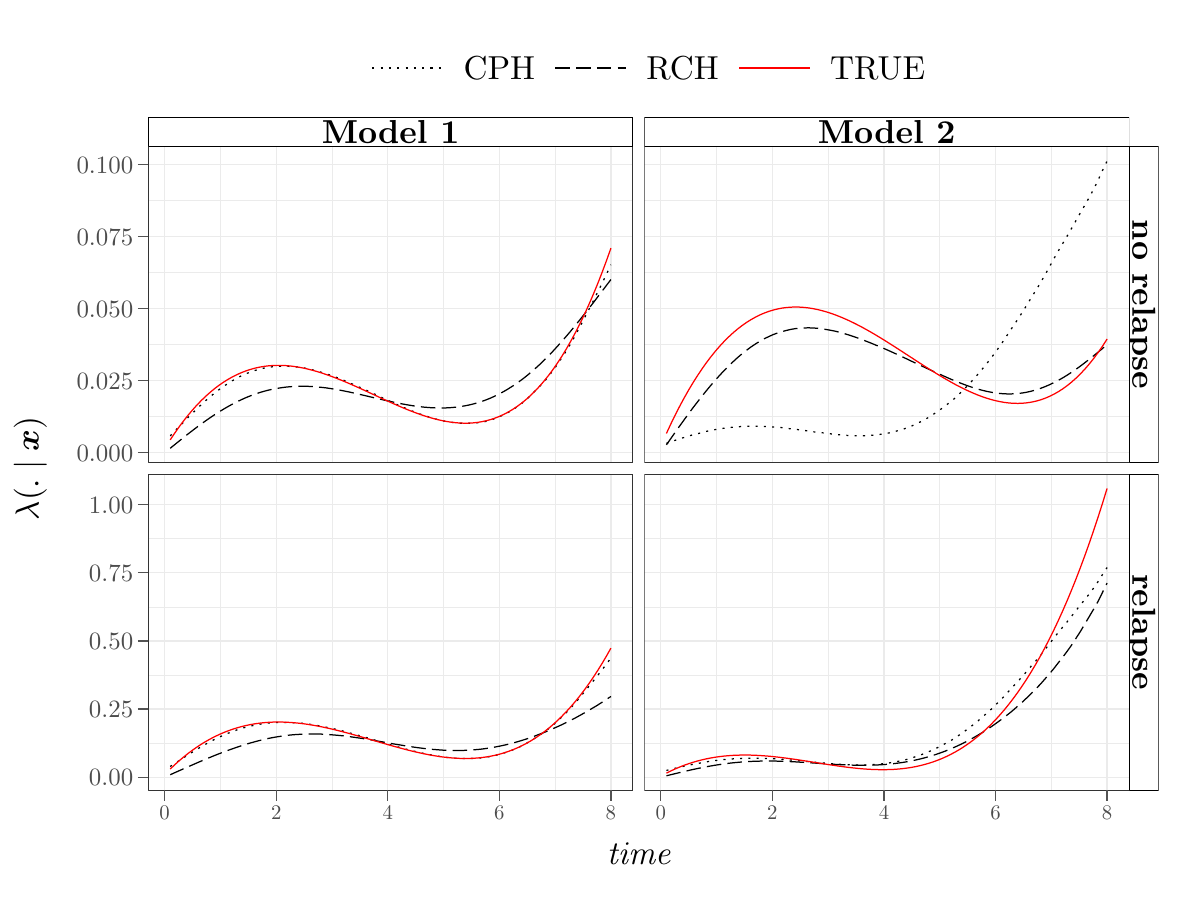}\vspace{-.2cm}
    \caption{\small Estimates for the conditional hazard with CPH and RCH  for an individual with the average age, median number of nodes, a small tumor,  in absence  (top) or presence  (bottom) of relapse, in Model 1 (left) and Model 2 (right). The plotted curves are the means over $500$ replications, with sample size $n=2000$ and    $20\%$ censoring. The true conditional hazard is the continuous line. \vspace{-.5cm}} 
    \label{fig:simu-comp-hazard}
\end{figure}

%\FloatBarrier

%----------------------------------------------
\subsection{Real data analysis}\label{sec:chap1-real-data}
%----------------------------------------------
% !TeX root = ../hazard.tex

We consider the Rotterdam primary breast cancer data set as considered by \cite{RA2013}; see also \cite{FPLPSKBNJGHS2000}.
The study includes 2982 patients treated for primary breast cancer following a surgery between 1978 and 1993. The event of interest is death due to breast cancer, with 1272 patients experiencing this outcome, resulting in a censoring rate of 57\%. Patients were followed for an average of 7.13 years, with visits every 3 to 6 months for the first 5 years and annually thereafter. The longest follow-up time was 19.5 years. Similarly to the empirical studies, we retained age, the number of positive lymph nodes, tumor size, and relapse occurrence as covariates. In this analysis the variable age is considered continuous, but the number of nodes is preserved as a discrete variable. Menopausal status, tumor grade, progesterone and \oe{}strogen receptor status, and treatment were excluded as they were not statistically significant in the Cox model estimated in Section \ref{sec:simu-implementation}. In the data, 48.5\% of patients had no positive lymph nodes, 12.3\% had one, and 8\% had two, with the fraction decreasing as the number of nodes increased. Tumor sizes were categorized as small (46.5\%), medium (10.2\%), and large (43.3\%). The median age at surgery was 54 years. The 25th (resp. 75th) quantile for the age is 45 (resp. 65) years old. Half of the patients in the sample have experienced the recurrence of the cancer during their follow-up.

Figure \ref{fig:app-surv} displays the estimated conditional survival functions for patients with various characteristics: no lymph nodes, small or large tumors, with or without relapse. The different lines correspond to different  ages at surgery. We compare estimates from the RCH method, with solid lines, and Cox's model  with dashed lines. Notably, our method generally yields more optimistic estimates. For instance, for patients with small tumors and relapse, the CPH estimates a survival probability below 0.25 after 12 years, regardless of age. However, among the 306 patients with small tumors, relapse, and no lymph nodes, 140 were censored after a median follow-up of 9 years, suggesting better survival than predicted by the CPH.

For large tumors, our method remains more optimistic, though the gap between RCH and CPH narrows, particularly for younger patients experiencing relapse. Additionally, the RCH estimates show that age influences survival outcomes differently across ages. For example, patients operated at age 40 with no relapse or lymph nodes have a near-zero estimated risk of dying from their cancer on the considered period according to the RCH which aligns with the data. Only two patients under 45 with no relapse died from breast cancer across tumor sizes, while 198 survived it. In contrast, the CPH predicts on average a 10\%–20\% share of these patients to die from breast cancer, despite their longer median follow-up of 9 years.
Conversely, younger patients who experience relapse have worse survival odds than older patients. In such cases, RCH and CPH estimates are similar for young patients but diverge for older patients. For instance, older patients (60–70 years) with relapse and large tumors have better survival estimates under the RCH than younger patients, while the CPH imposes worse odds for older patients due to the positive age coefficient. This flexibility in the RCH allows  to capture age- and relapse-dependent survival dynamics more accurately than the CPH. Additionally, reasons why the PH assumption is often unrealistic for cancer data can be found in \cite{SRL2007}. One situation where the RCH is less optimistic that the CPH is for the estimated survival probability for 70-year-old patients with large tumors and no relapse, with 72.5\% for the RCH and 80\% for the CPH. This is one of the few cases where the RCH produces lower estimates than the CPH. 
Overall, the most discriminant covariate for survival is the occurrence of relapse. Understanding the factors that lead to such relapse  requires a completely separate time-to-event analysis beyond our scope.

\begin{figure}[ht]
    \centering\vspace{-.2cm}
    \includegraphics[width=0.8\linewidth, height=8.79cm]{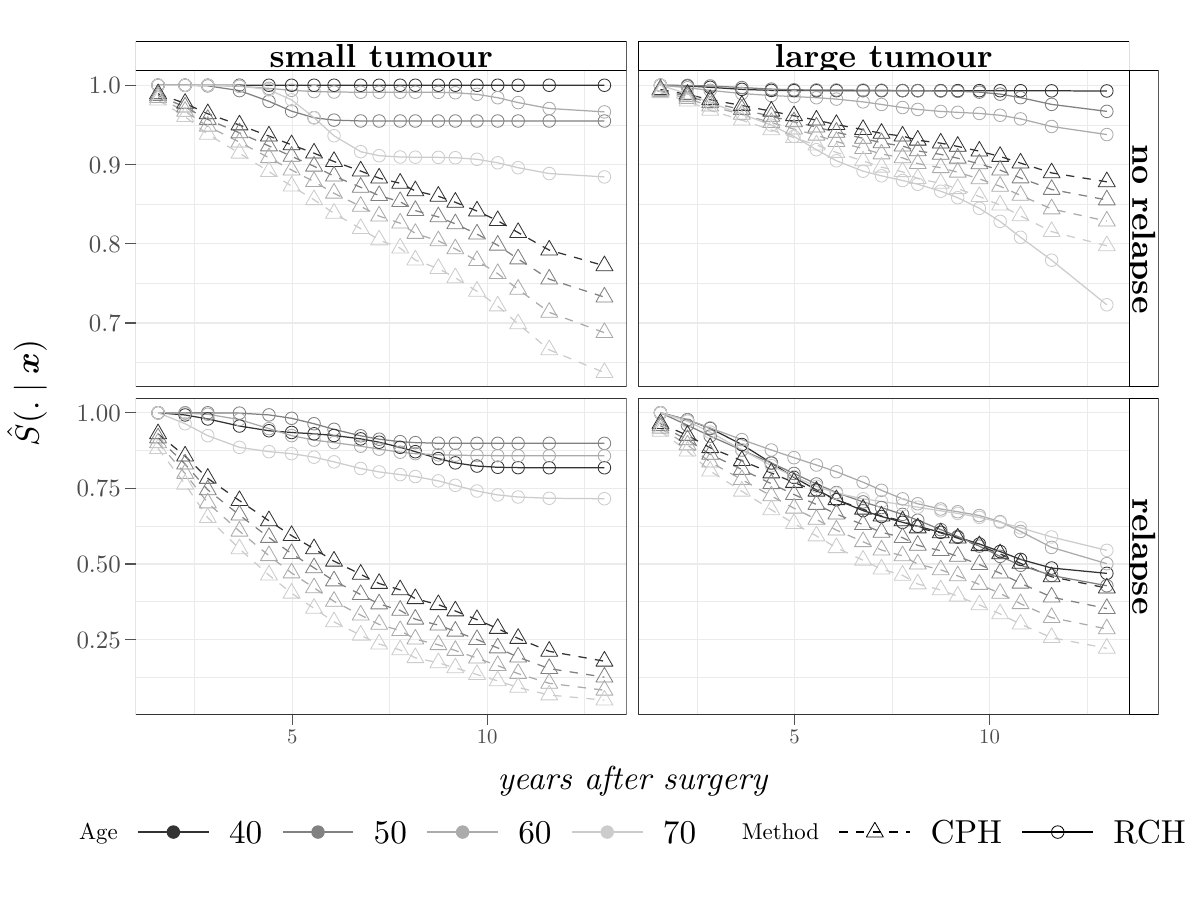}\vspace{-.3cm}
    \caption{\small Estimates of the conditional survival function of primary breast cancer patients. The estimates are produced by CPH and RCH methods and are displayed for patients with no node, and different size of tumor, age, and situation of relapse.}
    \label{fig:app-surv}
\end{figure}

%\FloatBarrier

%------------------------------------------
\section{Conclusions}\label{sec:chap1-ccl}
%------------------------------------------
% !TeX root = ../hazard.tex

In this paper, we propose a new, general perspective on nonparametric regression for time-to-event data subject to censoring and/or truncation, in the presence of competing risks, allowing for cured individuals, applicable to univariate or bivariate lifetime variables. The approach is based on a representation of the conditional hazard function given the predictors as a ratio between a joint density and a conditional expectation.  The joint density refers only to individuals for whom the time-to-event is observed. The representation can be established with continuous and discrete predictors and requires only mild identification conditions. 

%For example, in the context of left or right censoring, the time-to-event needs to be conditionally independent of the censoring given the predictors.  

%\begin{figure}[ht]
%	\centering\vspace{-.2cm}
%	\includegraphics[width=0.8\linewidth, height=8.79cm]{figures/APPLICATION/application-surv-chart}\vspace{-.3cm}
%	\caption{\small Estimates of the conditional survival function of primary breast cancer patients. The estimates are produced by CPH and RCH methods and are displayed for patients with no node, and different size of tumor, age, and situation of relapse.}
%	\label{fig:app-surv}
%\end{figure}

The numerator and denominator in the conditional hazard representation can be estimated using many existing nonparametric approaches. We propose a class of simple, one-pass kernel estimators that can be used to analyze and easily update massive datasets, such as those available via online streams. We show that the hazard function estimator achieves optimal convergence rates. The optimal bandwidth for the pointwise quadratic risk is proportional to the optimal bandwidth of a standard Parzen-Rosenblatt  estimator. A simulation experiment in the case of random right censoring shows that our estimator performs well.

Many issues have yet to be addressed. For example, we have assumed that the regularities of the joint density and the conditional expectation involved in representing the conditional hazard are the same and known. Adaptive methods would be preferable, although they are likely to be more challenging in one-pass, with continuously updated datasets. Another challenge relates to the case where the existence of a simple representation for the conditional hazard of the time-to-event variable is not clear, as is the case with the current status data. %It may be interesting to note that the model characterized by the equations \eqref{model_eq1PR} reaches the current status setup at the limit when $p(\x)$ tends to zero. Would this help to work out an approximation of the conditional distribution of interest in the current status setup? We leave the investigation of these issues to future work. 

%------------------------------------------
\appendix
\section{Appendix} 
In the Section \ref{app1} below we provide details on the representation of the conditional hazard function in the case of left truncation and right censoring with a cure rate. In Section \ref{app_proofs} we provide the technical details for justifying the theoretical statements given in Section \ref{sec:chap1-model-theory}.

% !TeX root = ../hazard.tex

\subsection{Cure models under left truncation}

Here, we provide details for the formula \eqref{eq:haz_108}. Recall that with a LTRC mechanism, there are latent variables $(T^*,C^*,\X^*, L^*)$ such that~: if $L^*_l > Y^*_l:=T_l^*\wedge C_l^*$, the $l-$th realization in the latent sample is not observed, while if $L^*_l \leq  Y_l^*$ the realization $(Y_l^*,\delta_l^*,\X_l^*, L^*_l)$, with $\delta_l^*=\1\{T_l^*\leq C_l^*\}$,  becomes the $i-$th observation  $(Y_i,\delta_i, \X_i, L_i)$ of the observed variables $(Y,\delta, \X, L)$. The observed sample is the result of an independent  sample $(T_l^*,C_l^*,\X_l^*, L^*_l)$, $l\geq1$, of the latent variables $(T^*,C^*,\X^*,L^*)$. Like in Section \ref{sec:chap1-model-discussion},  assume that 
%\begin{equation}\label{cure_LTRC}
$T^* \perp L^* \perp  C^* \mid \X^*$
 and 
 $\PP(L^* \leq  T^*\wedge C^*\mid \X=\x) \geq c>0$, 
 for some constant $c$.
%\end{equation}
Here, we assume   $\PP (T^*=\infty \mid \X=\x) = 1-\phi(\x)>0$. Then, we have 
\begin{equation}\label{eq_tstar} 
S_{T^*\mid \X^*}(t\mid \x) 	=1-\phi(\x)+\phi(\x) S_{0,T^*\mid \X^*}(t\mid \x)  ,
\end{equation}
where $S_{0,T^*\mid \X^*}(t\mid \x)=\PP(T^*>t\mid T^*<\infty , \X^*=\x)$ is the (proper) survival function of the latent susceptibles. Let $[0,\tau(\x)]$ be the bounded support of $S_{0,T^*\mid \X^*}$ and assume $\inf_{\x}\inf_{t\in [0,\tau(\x)]}S_{C^*\mid \X^*}(t\mid \x)>0$. Following the lines of the calculations leading to Equation (6.4.5) in \cite{E2018} with a hidden (latent) duration of interest $T^*\in[0,\infty]$ and the truncation variable $L^*$,  and using the expression of $\lambda_{0,T^*\mid \X^*}$ given in \citep[p. 2327]{PVK2020}, it can be shown that
$$
f_{Y,\delta,\X}(t,1,\x) =  \lambda_{0,T^*\mid \X^*}(t\mid \x) \times F_{L^*}(t\mid \x) p(\x)^{-1} \phi(\x)S_{0,T^*\mid \X^*}(t\mid \x) S_{C^*\mid \X^*}(t\mid \x)f_\X(\x),
$$
where $\lambda_{0,T^*\mid \X^*}$ is the conditional hazard associated to $S_{0,T^*\mid \X^*}$ and, with obvious notation, 
$$
p(\x) = \PP(L^* \leq T^*\wedge C^*\mid \X=\x)= \int_0^\infty S_{T^*\mid \X^*}(t\mid \x) S_{C^*\mid \X^*}(t\mid \x) F_{L^*\mid \X^*}(dt\mid \x) .
$$
On the other hand, by the definition of the observation mechanism, 
\begin{multline}
\mathbb P (L\leq t \leq Y \mid \X=\x)= 	 F_{L^*}(t\mid \x)   p(\x)^{-1}\phi(\x) S_{0,T^*\mid \X^*}(t\mid \x) S_{C^*\mid \X^*}(t-\mid \x)\\  + F_{L^*}(t\mid \x)   p(\x)^{-1} \{1-\phi(\x)\}S_{C^*\mid \X^*}(t-\mid \x). 
\end{multline}
To derive \eqref{eq:haz_108}, it remains to notice that, by the identification assumptions and the definition the observation mechanism, we have
\begin{multline}
	 F_{L^*}(t\mid \x) p(\x)^{-1}S_{C^*\mid \X^*}(t-\mid \x) =p(\x)^{-1}\mathbb P (L^*\leq t \leq C^* \mid \X=\x)\\
  =p(\x)^{-1}\mathbb P (L^*\leq t \leq C^* , L^*\leq Y^*\mid \X=\x)
	\\=\mathbb P (L^*\leq t \leq C^* \mid  L^*\leq Y^*,\X=\x)  = \mathbb P (L\leq t \leq C \mid   \X=\x).
\end{multline}

\label{app1}
% !TeX root = ../hazard.tex

\subsection{Technical proofs}
Here, we provide the proofs of the statements made in Section \ref{sec:chap1-model-theory}. Below, the symbol $\sim$ means the left side is bounded above and below by constants times the right side. $\overline C$, $\mathfrak C $, $\mathfrak c, \mathfrak c^\prime $ are constants, possibly different from line to line, but not depending on $n$, $t$, $\z$ or $c$ (the bandwidth factor).  
 
%\medskip

\begin{proof}[Proof of Proposition \ref{stoch1}]
Let $\Omega_n=\sum_{j=1}^n j^{\beta}$ with $\beta>0$ from the definition of the weights $\omega_{n,i}$, such that $\omega_{n,i}=i^{\beta}/\Omega_n$. Since $\beta<1$,  by Lemma \ref{serie_Riemann}, we have $\Omega_n\sim n^{1+\beta}$.	For any $\z\in \mathcal Z$,   $t\in \mathcal T$, and $\x_d$ in the finite set $\mathcal X_d$,  let
\begin{equation}\label{generic_app_b}
	G_i(t,\z,\x_d;c)=   a_{i} D_i(t) \mathbb K  \left(\mathfrak h(i)^{-1}(\Z_i-\z)  \right)  \1 \{  \X_{i,d} = \x_d\}, \quad \text{ with } \; a_{i} = \{\omega_{n,i}\Omega_n\}\mathfrak h(i)^{-p},
\end{equation}	 
where $c$ is the constant from the bandwidths $\mathfrak h(i)$.
Let us recall the notation~: for the case $\widehat{g}_{n}=\widehat f_n$, we denote  $\Z_i = (Y_i,\X_{c,i})$, $\z =(t,\x_c)$ and  $\Z_i,\z \in \mathbb R^{p}$, with $p=d_c+1$; for   the case $\widehat{g}_{n}=\widehat R_n$, we denote  $\Z_i = \X_{c,i}$, $\z =\x_c$ and  $\Z_i,\z \in \mathbb R^{p}$, with $p=d_c$. In view of the  definition of $G_i$, we write $ \widehat{g}_{n}(t,\x;c)  $ instead of $\widehat{g}_{n}(t,\x)$, and we have $ \widehat{g}_{n}(t,\x;c) = \Omega_n^{-1}  \sum_{i=1}^n G_i(t,\z,\x_d;c)$.

By our assumptions, $a_i\sim i^{\alpha p +\beta }$.   
We want to apply Bernstein's inequality  which states that, for any   $X_1,\ldots, X_n$  zero mean independent variables, such that $|X_i|\leq M_n$ for all $1\leq i\leq n$, we have
$$
\forall v\geq 0, \quad \PP \left(  \left|\sum_{i=1}^n X_i\right| \geq v\right)\leq 2 \exp\left( - \frac{v^2/2}{\sigma^2_n+M_nv/3}\right),\quad \text{ where } \quad \sigma_n^2 = \sum_{i=1}^n \EE (X_i^2).
$$
See, for example, \citep[Theorem 2.8.4]{V2018}. For an arbitrary $\x_d\in\mathcal X_d$, let 
\begin{equation}\label{xi_def}
X_i = X_i(t,\z,\x_d;c) = G_i(t,\z,\x_d;c) - \EE[G_i(t,\z,\x_d;c)],
\end{equation}
which implies  $M_n \sim \max\{a_1,\ldots ,a_n\} \sim n^{\alpha p +\beta }$. By Lemma \ref{serie_Riemann} and  since the expectation of $D(t)$ given $\Z = \z$ and $\X_d=\x_d$ is bounded and bounded away from zero, we have 
\begin{equation}\label{sigma_rate}
\sigma^2_n \sim \sum_{i=1}^n i^{2(\alpha p +\beta) }\EE [D(t)\mathbb K ^2 \left(\mathfrak h(i)^{-1}(\Z_i-\z)  \right) \1 \{  \X_{i,d} = \x_d\}] 
%\\= \sum_{\x_d\in \mathcal X_d}\PP (\X_{i,d} = \x_d)\sum_{i=1}^n i^{2(\alpha p +\beta) }\EE [\mathbb K ^2 \left(\mathfrak h(i)^{-1}(\Z_i-\z)  \right) \mid  \X_{i,d} = \x_d] 
 \sim n^{\alpha p +2\beta+1 }.
\end{equation}
Taking $v= \overline C (\log n)^{1/2} n^{(\alpha p +1)/2 + \beta}$ with $\overline C$ some large constant, since $(\alpha p +1)/2< 1$, we have
$$
\frac{v^2/2}{\sigma^2_n+M_nv/3} \sim \frac{\overline C^2 n^{(\alpha p +1) + 2 \beta}\times \log n}{n^{\alpha p +2\beta+1 } + n^{\alpha p +\beta } \times \overline C  (\log n)^{1/2}n^{(\alpha p +1)/2 + \beta} }\sim \overline C \log n.
$$
Noting that $v\Omega_n^{-1} \sim  (\log n)^{1/2} n^{(\alpha p -1)/2 }$, we deduce that, for any $\x_d\in\mathcal X_d$,
\begin{multline}
%	\PP \left( \left|\widehat{\mathfrak g}_{n}(t,\z,\x_d) - \EE\left[ \widehat{\mathfrak g}_{n}(t,\z,\x_d)\right]\right| \geq C (\log n)^{1/2} n^{(\alpha p -1)/2 } \right) 
	\PP \left(\left|\widehat{g}_{n}(t,\x;c)- \EE[\widehat{g}_{n}(t,\x;c)] \right| \geq\overline C (\log n)^{1/2} n^{(\alpha p -1)/2 } \right) 
		\\=	\PP \left( \Omega_n^{-1}  \left|\sum_{i=1}^n X_i(t,\z,\x_d;c)\right| \geq \overline C (\log n)^{1/2} n^{(\alpha p -1)/2 } \right) 
	\leq 2 \exp(-\mathfrak c \times \overline C \log n),
\end{multline}
with $\mathfrak c$ a constant depending only on $\alpha,\beta$, the  $\underline c,\overline c$ giving the bandwidth range in condition \ref{ker1} of Assumption \ref{assump1}, and  the uniform norm of $\mathbb K$.   The pointwise rate from Proposition \ref{stoch1} follows.

To obtain the uniformity with respect to $t$, $\z$, and the constant $c$ in the bandwidth, we consider an equidistant grid of $\mathcal G_1$ points $(\z_{\ell^\prime }, c_{\ell^{\prime\prime}})$ in the set $ \Z\times [\underline c , \overline c]$, with $\mathcal G_1\sim n^{\gamma_1}$, $\gamma_1 \geq 2$. Using the Lipschitz conditions on the kernel $ K(\cdot)$, for any $\z,c$, indices $\ell^\prime, \ell^{\prime\prime}$ exist such that 
$$
|G_i(t,\z,\x_d;c) - G_i(t,\z_{\ell^\prime },\x_d; c_{\ell^{\prime\prime}})|\leq \mathfrak C n^{-1}, \qquad \forall n\geq 1, \forall 1\leq i \leq n,
$$ 
where $\mathfrak C$ is some constant.  Next, for the case where the $G_i$'s correspond to the estimator of $R(t,\x)$,  consider an equidistant ordered grid of $\mathcal G_2$ points $t_{\ell}$ on the compact interval $\mathcal T$, with $\mathcal G_2\sim n^{\gamma_2}$, $\gamma_2 \geq 2$. Then, by Assumption \ref{assump1}-\ref{xia2c}, the monotonicity of  $D(\cdot)$ guarantees that an index $\ell$ exists such that 
\begin{equation}\label{ineqE1}
G_i(t_{\ell+1 },\z_{\ell^\prime },\x_d; c_{\ell^{\prime\prime}}) \leq G_i(t,\z,\x_d;c) \leq G_i(t_{\ell },\z_{\ell^\prime },\x_d; c_{\ell^{\prime\prime}}).
\end{equation}
and, by the Lipschitz condition on the kernel, 
\begin{equation}\label{ineqE2}
\left| \EE[G_i(t_{\ell + 1},\z_{\ell^\prime },\x_d; c_{\ell^{\prime\prime}})] -\EE[G_i(t_{\ell },\z_{\ell^\prime },\x_d; c_{\ell^{\prime\prime}})] \right|\leq \mathfrak C n^{-1}, \qquad \forall n\geq 1, \forall 1\leq i \leq n,
\end{equation}
where $\mathfrak C$ is some constant. Gathering facts, for any $t$, $\z$, and the constant $c$ in the bandwidth, indices  $\ell$, $\ell^\prime$, $\ell^{\prime\prime}$ exist such that 
\begin{multline}
%\left\{G_i(t_{\ell },\z_{\ell^\prime },\x_d; c_{\ell^{\prime\prime}}) - \EE[ G_i(t_{\ell },\z_{\ell^\prime },\x_d; c_{\ell^{\prime\prime}})] \right\}\\\hspace{-3cm} 
X_i(t_{\ell +1},\z_{\ell^\prime },\x_d; c_{\ell^{\prime\prime}})+ 
		 \left\{\EE[ G_i(t_{\ell +1},\z_{\ell^\prime },\x_d; c_{\ell^{\prime\prime}})] - \EE[ G_i(t,\z,\x_d;c)]\right\}\\
\hspace{-4cm}	\leq  G_i(t,\z,\x_d;c) - \EE[ G_i(t,\z,\x_d;c)]
%	\\ \hspace{3cm} \leq  \left\{G_i(t_{\ell },\z_{\ell^\prime },\x_d; c_{\ell^{\prime\prime}}) - \EE[ G_i(t_{\ell },\z_{\ell^\prime },\x_d; c_{\ell^{\prime\prime}})] \right\}\\
\\\leq X_i(t_{\ell },\z_{\ell^\prime },\x_d; c_{\ell^{\prime\prime}}) 	+\left\{	\EE[ G_i(t_{\ell },\z_{\ell^\prime },\x_d; c_{\ell^{\prime\prime}})] - \EE[ G_i(t,\z,\x_d;c)]\right\}.
\end{multline}
From this \eqref{ineqE1}, \eqref{ineqE2} and Boole's union bound inequality, we deduce that constants $\mathfrak c, \mathfrak c^\prime, \mathfrak C$ exist such that, 
\begin{multline}
	\PP \left(\sup_{t,\x,c}\left|\widehat{g}_{n}(t,\x;c)- \EE[\widehat{g}_{n}(t,\x;c)] \right| \geq \overline  C (\log n)^{1/2} n^{(\alpha p -1)/2 } \right) 
%\\ \leq 	\PP \left(\sup_{\ell,\ell^\prime , \ell^{\prime\prime} }\left|\widehat{g}_{n}(t_\ell,\z_{\ell^\prime},\x_d;c_{\ell^{\prime\prime}})- \EE[\widehat{g}_{n}(t_\ell,\z_{\ell^\prime},\x_d;c_{\ell^{\prime\prime}})] \right| \geq \overline  C (\log n)^{1/2} n^{(\alpha p -1)/2 } + \mathfrak Cn^{-1}\right)
%	\\ \leq \sum_{\ell,\ell^\prime,\ell^{\prime\prime}}	\PP \left(\left|\widehat{g}_{n}(t_\ell,\z_{\ell^\prime},\x_d;c_{\ell^{\prime\prime}})- \EE[\widehat{g}_{n}(t_\ell,\z_{\ell^\prime},\x_d;c_{\ell^{\prime\prime}})] \right| \geq (\overline C+1) (\log n)^{1/2} n^{(\alpha p -1)/2 } \right)
	\\	\leq 2|\mathcal X_d| \mathcal G_1 \mathcal G_2 \exp(-\mathfrak c \times \overline C \log n)\\ \leq 2|\mathcal X_d|\exp\{\mathfrak c^\prime  (\gamma_1+\gamma_2)  \log n-\mathfrak c  \overline  C \log n\}\rightarrow 0.
\end{multline}
Here, $|\mathcal X_d|$ denotes the cardinal of $\mathcal X_d$. The  convergence to zero is guaranteed by the fact that the set of discrete values $\x_d$ is finite and $\overline C$ can be large. This proves the result. \end{proof} 

%\medskip

\begin{proof}[Proof of Proposition \ref{bias1}]
Let us recall that 
$G(t,\x)= f_{Y,\delta,\X}(t, 1, \x)$ or $G(t,\x)=R(t, \x)$, 
depending on the whether $\widehat g_n$ estimates the numerator or denominator in \eqref{eq_rec_haza}. 
We only consider the case where $\widehat g_n$ estimates the numerator $f_{Y,\delta,\X}(t, 1, \x)$, the other case being very similar. Since for the numerator case the smoothing is taken with respect to $Y$ and $\X_c$, $\Z = (Y,\X_{c})$, $\z =(t,\x_c)$, $p=d_c+1$, we here write $G(\z,\x_d)$ instead of $G(t,\x)$.
Then,   $\forall \x\in\mathcal X_\epsilon$ and $\forall t\in \mathcal T$, $t\geq \epsilon$,  by the definitions, a change of variables and the second order Taylor expansion,   we get 
\begin{multline}\label{T_E}
E_i(t,\x;c):= \EE[ \1\{T\leq C\} \mathbb K_{\mathfrak h(i)} (\Z-\z  )\1 \{  \X_{d} = \x_d\}] 
\\ = \int _{\mathbb R^p} G(\mathbf w,\x_d ) \mathbb K_{\mathfrak h(i)} (\mathbf w -\z  )  d\mathbf w 
=  \int _{\mathbb R^p}  G(\z+ \mathfrak h(i)\mathbf u, \x_d)\mathbb K (\mathbf u  ) d\mathbf u \\= G(\z, \x_d) + \frac{\mathfrak h(i)^2}{2} \mu_2 (\mathbb K) \operatorname{Trace}(\mathcal H_G(\z,\x_d)) \{1+o(1)\},
\end{multline}
provided $i$ is sufficiently large. Here, $\mathcal H_G(\z,\x_d)$ denotes the Hessian matrix of $G$ with respect to $\z$,  which, by our assumptions, is uniformly bounded. 
The calculations in \eqref{T_E} are valid as soon as $i$ is sufficiently large such that $\z+ \mathfrak h(i)\mathbf u\in\mathcal T\times \mathcal X_c$. 
Next, since 
$$
\EE \left[\widehat{g}_{n}(t,\x;c)\right]= \sum_{i=1}^n \omega_{n,i} E_i(t,\x;c),  
$$
and, by Lemma \ref{serie_Riemann}, $\sum_{i=1}^n \omega_{n,i}\mathfrak h(i)^2\sim n^{\beta + 1 -2\alpha}$, the result follows. 

It worth noting that the proof remains valid even when $\epsilon $ decreases to zero at a suitable rate. Such a relaxation matters for estimating the numerator in the conditional hazard function representation, in which case the dimension of $\Z$ is $p\geq 2$. Indeed, let $i_0$ such the calculations in \eqref{T_E} are valid as soon as $i\geq i_0$. Then the bias is of order 
$$
\left[ \sum_{i=1}^{n} \omega_{n,i}\right]^{-1}\left\{ \sum_{i=1}^{i_0} \omega_{n,i}\mathfrak h(i)+ \sum_{i=i_0}^n \omega_{n,i}\mathfrak h(i)^2  \right\} \sim n^{-(\beta + 1) }\left[  i_0^{\beta + 1-\alpha} + \{n^{\beta + 1-2\alpha} - i_0^{\beta + 1-2\alpha}\} \right],
$$ 
and it suffices to guarantee that $ i_0^{\beta + 1-\alpha} / n^{\beta + 1-2\alpha} \rightarrow 0$ in order to preserve the same uniform rate for the bias. Then we can take $ i_0 \sim n^\nu $, with   $0\leq \nu< (\beta+1-2\alpha)/ (\beta+1-\alpha)$, $\epsilon=\epsilon_n \sim n^{-\alpha\nu}$. Note that with $p\geq 2$, our assumptions guarantee that $\beta+1-2\alpha>0$. We can also notice that taking $\beta >0$ allows for larger $\nu$ and thus smaller $\epsilon_n$. \end{proof}

%\medskip

\begin{lemma}\label{serie_Riemann}
	Let  $\rho \geq 0$, $\varrho >0$, $\varrho \neq 1$, 
	and $n>1$. Then
\begin{equation}\label{sumR1}
	\frac{(n+1)^{1-\varrho}-1}{1-\varrho} 
	% - \frac{1}{1-\varrho}	
	\leq \sum_{i=1}^n i^{-\varrho} 
	\leq \frac{n^{1-\varrho}-\varrho}{1-\varrho} 
	%-\frac{\varrho}{1-\varrho},
%\end{equation}
\quad \text{	and }\quad
%\begin{equation}\label{sumR2}
	\frac{n^{1+\rho}+\rho}{1+\rho} \leq \sum_{i=1}^n i^{\rho } \leq \frac{(n+1)^{1+\rho}-1}{1+\rho}.
\end{equation}
	
\end{lemma}

\begin{comment}

\medskip

\begin{proof}[Proof of Lemma \ref{serie_Riemann}]
	For any $i \geq 1$ and $x\in [i,i+1]$, we have
	$$
	(i+1)^{-\varrho} \leq x^{-\varrho} \leq i^{-\varrho}.
	$$	
	We then deduce 
	$$
	\sum_{i=2}^{n+1} i^{-\varrho} = \sum_{i=1}^{n} (i+1)^{-\varrho} \leq \int_{1}^{n+1} x^{-\varrho} dx = \frac{1}{1-\varrho}\left[(n+1)^{1-\varrho} - 1
	\right]\leq \sum_{i=1}^{n} i^{-\varrho},
	$$
	from which \eqref{sumR1} follows. 
	
	Since $\rho >0$, for any $i \geq 1$ and $x\in [i,i+1]$, we have
	$
	i^{\rho} \leq x^{\rho} \leq (i+1)^{\rho},
	$
	and thus
	$$
	\sum_{i=1}^{n} i^{\rho} \leq \int_{1}^{n+1} x^{\rho} dx = \frac{1}{1+\rho}\left[(n+1)^{1+\rho} - 1
	\right]\leq \sum_{i=1}^{n} (i+1)^{\rho} = \sum_{i=2}^{n+1} i^{\rho} ,
	$$
and  \eqref{sumR2} follows. The case $\rho=0$ is obvious. \end{proof}

\end{comment}
\label{app_proofs}
%\input{sections/appendix2}\label{app2}
%------------------------------------------

% Chapitre pour la bibliographie
% Bibliography chapter
% \nocite{*}
%\printbibliography

\section*{Acknowledgements}

Valentin Patilea gratefully acknowledges discussions with Jean-Marie Rolin. These discussions inspired the approach to survival analysis proposed in this paper.

\bibliographystyle{apalike}
\bibliography{references-chapitre1}

\end{document}